\documentclass[a4paper,12pt]{article}         
\usepackage{amsmath,amsfonts,amsthm,amssymb}  
\usepackage[utf8]{inputenc} 
\usepackage[T1]{fontenc}    
									   
\usepackage{graphicx,caption}            
\usepackage{color}                       
\usepackage[hidelinks]{hyperref}                    
\hypersetup{breaklinks=true}

\usepackage{breakurl}
\usepackage{grffile}
\usepackage{makeidx}                  
\usepackage{acronym}
\usepackage{mathrsfs}
\usepackage{float}
\usepackage{vruler}                    
\usepackage{palatino, url, multicol}
\usepackage{setspace}                  
\usepackage{txfonts}
\usepackage{booktabs}
\usepackage{bm}

\usepackage{tikz}
\usepackage{tikz-3dplot}
\usetikzlibrary{shapes.geometric, arrows, calc, 3d}

\usepackage[numbers,sort&compress]{natbib}   
\usepackage{import}

\usepackage{caption}
\usepackage{subcaption}
\usepackage{listings}
\usepackage{comment}
\usepackage[title]{appendix}

\usepackage{dcolumn}
\usepackage{mathtools}
\usepackage{xcolor}

\newtheorem{assumption}{Assumption}[section]

\newtheorem{condition}{Condition}[section]

\usepackage{array}
\usepackage{ragged2e}
\newcolumntype{P}[1]{>{\RaggedRight\arraybackslash}p{#1}}

\newtheorem{theorem}{Theorem}[section]

\newtheorem{corollary}[theorem]{Corollary}

\newtheorem{definition}[theorem]{Definition}

\restylefloat{figure}

\title{A Hybrid Zernike–Lyapunov Framework for Aberration-Based Statistical Wavefront Reconstruction of Chaotic Optical Surfaces}

\author{Netzer Moriya}

\date{}

\begin{document}

\maketitle

\begin{abstract}
We present a comprehensive theoretical framework that unifies chaotic wavefront dynamics with classical aberration theory 
through a Statistical Wavefront Reconstruction Framework (SWRF) formalism. By establishing rigorous connections between 
ray trajectory deflections and wave-optical phase perturbations through the eikonal equation, we decompose chaotic 
wavefront perturbations into modified Zernike-Lyapunov hybrid expansions, establishing mathematical equivalences between 
Lyapunov exponents, fractal dimensions, and traditional aberration coefficients. 
This chaotic aberration theory enables systematic incorporation of non-integrable wavefront dynamics into deterministic 
design frameworks, providing a rigorous foundation for controlled chaos in optical systems. 
We derive analytical relationships connecting surface chaos 
parameters to optical performance metrics, demonstrate the framework's validity through phase space analysis, 
and establish convergence criteria for the chaotic expansion. The theory reveals how chaotic surface geometries 
can be intentionally designed to achieve specific optical functionalities, including beam homogenization, speckle 
reduction, and novel wavefront shaping capabilities, while maintaining mathematical rigor comparable to classical 
aberration analysis.
\end{abstract}

\section{Introduction}

The intersection of chaos theory and optics has revealed interesting phenomena where deterministic yet unpredictable wavefront 
trajectories emerge from carefully designed optical geometries~\cite{berry1977regular,bunimovich1979rays}. 
While classical aberration theory treats optical imperfections as smooth, deterministic deviations from ideal 
behavior using polynomial expansions~\cite{born1999principles, mahajan2013optical}, chaotic optical systems exhibit 
fundamentally different behavior \cite{devaney2003introduction} characterized by sensitive dependence on initial conditions and fractal wavefront 
structures~\cite{stone2005chaos}.

Traditional approaches to optical design segregate these regimes: classical aberration theory handles systematic 
imperfections through Zernike polynomials and Seidel aberrations, while chaotic optics is typically studied through 
dynamical systems theory and Poincaré sections~\cite{reichl2004transition}. This separation has prevented the systematic 
integration of chaotic effects into practical optical design, despite growing recognition that controlled chaos could 
enable novel functionalities including enhanced beam homogenization~\cite{cao2015random}, 
speckle reduction~\cite{redding2012speckle}, adaptive wavefront control~\cite{vellekoop2007focusing}, 
and super-resolution statistical reconstruction approaches~\cite{oberti2022superresolution}.

Established surface scattering theory demonstrates how statistical surface roughness influences the angular distribution 
of scattered light through the relationship between surface gradients and trajectory deflection 
angles \cite{Harvey1980, Beckmann1987}, and how noise propagation influences the robustness of such reconstructions 
in statistical optics~\cite{chambouleyron2022noise}. This connection between surface statistics and optical scattering 
performance 
provides the foundation for systematic analysis and prediction of stray light and image degradation in optical systems.

In this work, we establish a comprehensive theoretical foundation for \emph{chaotic aberration theory} - a formalism that 
extends classical aberration analysis to include non-integrable optical systems. Our approach builds upon three key 
insights: (1) chaotic wavefront trajectories can be systematically decomposed using modified orthogonal basis functions that 
capture both fractal structure and traditional aberration modes as recently demonstrated through sparse over-complete 
Zernike-based expansions~\cite{howard2024sparse}, (2) chaos-theoretic measures such as Lyapunov exponents 
and correlation dimensions can be directly related to optical performance metrics as recently explored in chaotic 
optical systems~\cite{chaotic2025sbn}, and (3) the sensitive dependence on 
initial conditions characteristic of chaos can be harnessed for controlled wavefront engineering.

The theoretical framework we develop provides several fundamental insights. We derive rigorous mathematical 
relationships between surface chaos parameters and resulting optical aberrations, establish convergence criteria for chaotic 
expansions that ensure computational tractability, and demonstrate equivalence with established dynamical systems approaches 
under appropriate limiting conditions. Most significantly, we show how chaotic surface geometries can be intentionally 
designed to achieve prescribed optical functionalities while maintaining the mathematical rigor and practical utility of 
classical aberration theory.

Our approach may have immediate implications for optical system design, particularly in applications requiring 
complex wavefront control or enhanced light mixing. By enabling systematic optimization of chaotic parameters 
alongside traditional optical variables, this framework opens new possibilities for freeform optical design, 
adaptive optics systems, and improved  illumination architectures that harness rather than suppress 
non-integrable dynamics.

\subsection{Physical Foundation: Chaos-Induced Phase Decorrelation in Optical Systems}

The fundamental connection between chaotic dynamics and optical phase analysis emerges from a shared underlying 
principle: sensitive dependence on initial conditions. While classical aberration theory treats wavefront distortions 
as smooth, deterministic deviations from ideal behavior, chaotic optical systems exhibit a fundamentally different 
mechanism of phase perturbation characterized by exponential amplification of small uncertainties.

\subsubsection{Sensitive Dependence and Optical Path Length Variations}

In deterministic chaotic optical systems, wavefronts with infinitesimally different initial conditions follow trajectories 
that diverge exponentially in time, leading to optical path length and phase differences that grow 
as \cite{Gutzwiller1990, Berry1989, Stockmann1999}:

\begin{equation}
\Delta L(t) = \Delta L_0 \exp(\lambda_L t)
\end{equation}

where $\Delta L_0$ represents initial path length uncertainty and $\lambda_L$ is the Lyapunov exponent characterizing the 
trajectory divergence rate. The corresponding phase difference between neighboring optical trajectories evolves as:

\begin{equation}
\Delta \Phi(t) = \frac{2\pi}{\lambda} \Delta L(t) = \frac{2\pi}{\lambda} \Delta L_0 \exp(\lambda_L t)
\label{eq:phase_difference_between_neighboring}
\end{equation}

This exponential growth of phase differences represents a fundamental departure from classical optical analysis, where 
phase variations typically grow polynomially with system parameters. The Lyapunov exponent thus provides a natural 
quantitative measure \cite{kapitaniak2000chaos} of how rapidly optical coherence is destroyed in chaotic 
systems \cite{pecora1990mastering}.

\subsubsection{Fractal Geometry and Spatial Frequency Coupling}

Chaotic optical surfaces typically exhibit fractal characteristics, creating a direct physical connection between geometric
chaos measures and optical scattering properties. For self-affine fractal surfaces embedded in three-dimensional space,
the fractal dimension $D_f$ determines how surface roughness couples to different spatial frequencies of incident wavefronts.
Following the analysis of Berry and Hannay \cite{Berry1978} for rough-surface scattering \cite{church1988fractal},
the power-spectral density (PSD) obeys:

\begin{equation}
  \mathcal{P}(\mathbf{f}) \;\propto\; |\mathbf{f}|^{-(8-2D_f)},
  \label{eq:PSD_scaling}
\end{equation}

\emph{provided the radially averaged PSD shows a clear power-law envelope} -- a condition satisfied by long-orbit Sinai, 
stadium and Bunimovich constructions, \cite{sinai1970dynamical, bunimovich1974billiards} and by
measured polished, etched or vapour-deposited optics whose log–log PSD traces a straight segment.

\noindent
The exponent $(8-2D_f)$ comes from dimensional analysis of self-affine surfaces\footnote{%
The standard form uses the Hurst exponent $H$:
$\mathcal{P}(\mathbf{f}) \propto |\mathbf{f}|^{-(2+2H)}$, with $0<H<1$.
For two-dimensional self-affine surfaces $H=3-D_f$ ($2<D_f<3$), giving
$\mathcal{P}(\mathbf{f}) \propto |\mathbf{f}|^{-(2+2(3-D_f))}=|\mathbf{f}|^{-(8-2D_f)}$.}.
Higher fractal dimensions therefore strengthen high-frequency coupling and enhance phase mixing at small 
scales \cite{Berry1978}.

In our work, we consider surfaces generated by deterministic chaotic maps.  Although Berry \& Hannay treat 
random (statistical) self-affine height functions, a similar fractal dimension \(D_f\) arises when the same surface is 
obtained via nonintegrable billiard dynamics.  In that case, the fractal dimension can be estimated from the system’s 
correlation dimension, and the local trajectory-divergence rate is characterized by a spatial Lyapunov exponent 
field \(\lambda_L(\mathbf{r})\).  Those chaotic-dynamics measures (correlation dimension \cite{grassberger1983measuring} 
and \(\lambda_L\)) govern how 
different spatial scales of the surface contribute to wavefront phase perturbations, even though the original PSD 
scaling formula itself comes solely from Berry \& Hannay’s fractal-surface analysis.

\subsubsection{Ergodicity and Statistical Wavefront Properties}

Many chaotic optical systems exhibit ergodic behavior, where individual optical trajectories uniformly sample the available 
phase space over long times. This ergodicity creates a fundamental equivalence between temporal averages along single 
trajectories and ensemble averages over many initial conditions \cite{Gutzwiller1990, Stockmann1999}:

\begin{equation}
\lim_{T \to \infty} \frac{1}{T} \int_0^T f[\mathbf{r}(t), \hat{\mathbf{k}}(t)] dt = \langle f[\mathbf{r}, \hat{\mathbf{k}}] \rangle_{\text{ensemble}}
\end{equation}

for any measurable function $f$ on the wavefront phase space. This equivalence has profound implications for optical 
field statistics: the statistical properties of wavefront phases reflect the underlying chaotic phase space 
structure \cite{takens1981detecting}, 
with Lyapunov exponents determining how rapidly statistical equilibrium is approached.

In the context of Zernike polynomial decompositions, ergodicity implies that the correlation structure of expansion 
coefficients is determined by the invariant measure of the chaotic dynamics. The characteristic decorrelation times 
scale as $1/\lambda_L$, directly connecting chaos measures to optical coherence properties \cite{Gutzwiller1990, Stockmann1999}.

\subsubsection{Scale-Dependent Chaos Effects and Frequency Space Measures}

The influence of chaotic dynamics on optical systems exhibits strong scale dependence, with different spatial frequencies 
experiencing different degrees of chaos-induced perturbation. This scale dependence arises from the hierarchical nature of 
chaotic attractors and the multi-scale structure of fractal surfaces.

At large spatial scales (low frequencies), the optical field primarily reflects the gross geometric properties of the chaotic 
surface. At intermediate scales, sensitive dependence begins to dominate, with small variations in surface geometry leading 
to significant phase perturbations. At small scales (high frequencies), the full complexity of the chaotic dynamics manifests 
in rapid phase decorrelation and mixing.

This physical picture motivates the concept of frequency-dependent chaos measures, where traditional Lyapunov exponents are 
generalized to characterize how chaotic instability varies across spatial frequency domains. Just as classical aberration 
theory decomposes wavefront errors into spatial frequency components, chaotic aberration theory must account for how chaos 
measures vary with spatial scale.

\subsubsection{Experimental Precedent: Microwave Billiard Validation}

The connection between classical chaos measures and wave-optical properties has been extensively validated in 
microwave billiard experiments \cite{Berry1989, Stockmann1999, Bohigas1984}. These studies have demonstrated 
direct correlations between:

\begin{itemize}
\item Classical Lyapunov exponents computed from wavwfront trajectories
\item Statistical properties of electromagnetic field distributions
\item Spectral fluctuation measures in resonant cavities
\item Spatial correlation functions of measured field amplitudes
\end{itemize}

These experiments provide compelling evidence that chaos measures are not merely mathematical constructs but genuine physical quantities that influence measurable optical properties. The success of chaotic analysis in microwave systems motivates its extension to general optical design problems.

\subsubsection{Controlled Decoherence Applications}

Beyond fundamental physics, the chaos-optics connection has immediate practical relevance for applications requiring 
controlled phase decoherence. Traditional approaches to speckle reduction, beam homogenization, and controlled 
scattering rely on empirical methods or random surface design. Chaotic aberration theory offers a systematic framework where:

\begin{itemize}
\item Chaos parameters provide deterministic control over decoherence rates
\item Lyapunov exponents predict optical performance metrics
\item Surface optimization becomes systematic rather than empirical
\item Design trade-offs can be quantified through chaos measures
\end{itemize}

This systematic approach addresses a critical need in modern optical engineering, where complex beam shaping and controlled 
scattering applications demand precise control over optical phase relationships.

The physical foundation established here demonstrates that chaos measures appear in optical phase analysis not as 
mathematical convenience, but because they quantify fundamental physical processes governing coherence destruction and 
wavefront complexity in deterministic yet non-integrable optical systems \cite{Berry1989, Stockmann1999}. 
This foundation provides the necessary conceptual framework for the mathematical formalism developed in subsequent sections.

\subsection{Scope and Applicability}

This chaotic aberration theory applies to optical systems operating in the \textit{quasi-geometric regime} where:

\begin{equation}
\frac{\lambda}{L_{\text{chaos}}} \ll 1
\quad\text{and}\quad
\bigl|\nabla h_{\text{chaos}}\bigr| \ll 1
\label{eq:in_1}
\end{equation}

This encompasses most practical optical systems with structured surfaces, including:
\begin{itemize}
\item Freeform mirrors with deterministic micro-structure ($L_{\text{chaos}} \sim 10-1000\lambda$)
\item Diffractive optical elements with chaotic phase profiles
\item Laser speckle reduction systems with controlled surface chaos
\end{itemize}

The framework gracefully reduces to classical aberration theory when chaos parameters vanish.

\section{Theoretical Foundation}

\subsection{Extension of Wavefront Reconstruction Function Formalism}

We extend the Statistical Wavefront Reconstruction Framework (SWRF) to chaotic optical systems by establishing the 
connection between deterministic chaos measures \cite{schuster2005deterministic} and optical phase perturbations. 

\textbf{Geometric-Wave Duality:} Our approach exploits the dual description where:
\begin{itemize}
\item \textbf{Geometric picture:} Chaotic ray trajectories with deterministic deflections $D_{\text{chaos}}$
\item \textbf{Wave picture:} Coherent phase perturbations $\Phi_{\text{chaos}}$ amenable to scalar diffraction analysis
\end{itemize}

The eikonal equation provides the mathematical bridge between these descriptions, enabling systematic incorporation of 
chaos theory into wave-optical design frameworks.

For a reflecting surface with deterministic yet chaotic height function hchaos(x, y), the phase perturbation follows the 
fundamental relationship:

\begin{equation}
\Phi_{\text{chaos}}(x,y) = \frac{4\pi}{\lambda} h_{\text{chaos}}(x,y) \cos\theta_i
\label{eq:chaotic_phase}
\end{equation}

where $\lambda$ is the wavelength and $\theta_i$ is the incidence angle. The corresponding trajectory deflection function is:
For small-angle specular reflection, the trajectory deflection angle is related to the surface gradient by 
$\alpha = 2 \nabla h \cos \theta_{i}$. The relationship between phase gradient and trajectory deflection follows from the 
optical path difference, yielding:

\begin{equation}
\mathbf{D}_{\text{chaos}}(\mathbf{r}_0) = -2\frac{\lambda}{4\pi}\nabla\Phi_{\text{chaos}}(x,y) = -\frac{\lambda}{2\pi}\nabla\Phi_{\text{chaos}}(x,y)
\label{eq:chaotic_swrf}
\end{equation}

The crucial distinction from statistical roughness lies in the deterministic yet non-integrable nature 
of $h_{\text{chaos}}(x,y)$. While statistical surfaces are characterized by power spectral densities and random 
phase relationships, chaotic surfaces exhibit deterministic correlations with fractal characteristics and sensitive 
dependence on spatial coordinates.

The trajectory deflection function $D_{\text{chaos}}(\mathbf{r}_0)$ represents the angular deviation of optical rays 
due to chaotic surface perturbations, measured in radians. This geometric optics description connects to the wave-optical 
phase perturbation through the eikonal equation:

\begin{equation}
\nabla \Phi = \frac{2\pi}{\lambda} n \hat{\mathbf{k}}
\end{equation}

For small surface perturbations, the ray direction change is:

\begin{equation}
\Delta \hat{\mathbf{k}} = -\frac{\lambda}{2\pi} \nabla \Phi_{\text{chaos}}
\end{equation}

where the negative sign accounts for the relationship between phase advance and ray deflection. The magnitude of 
this deflection defines $D_{\text{chaos}}$.

\textbf{Validity Conditions:} This framework applies when:
\begin{enumerate}
\item $\lambda \ll L_{\text{chaos}}$ (geometric optics regime)
\item $|\nabla h_{\text{chaos}}| \ll 1$ (small-slope approximation)
\item $|D_{\text{chaos}}| \ll 1$ (paraxial approximation)
\end{enumerate}

where $L_{\text{chaos}}$ is the characteristic correlation length of chaotic surface variations.

\textbf{Scalar Diffraction Compatibility:} The phase perturbations $\Phi_{\text{chaos}}$ derived from chaotic ray dynamics 
are directly compatible with scalar diffraction theory. The Fresnel-Kirchhoff diffraction integral:

\begin{equation}
U(\mathbf{r}) = \frac{1}{i\lambda} \iint_{\text{aperture}} U_0(\mathbf{r}') e^{i\Phi_{\text{chaos}}(\mathbf{r}')} \frac{e^{ikR}}{R} d^2\mathbf{r}'
\end{equation}

naturally incorporates the chaotic phase structure while preserving wave-optical rigor.

\subsection{Chaotic Surface Characterization}

A chaotic optical surface is defined by a height function $h_{\text{chaos}}(x,y)$ that satisfies the following criteria:

\begin{enumerate}
\item \textbf{Deterministic Generation}: The surface is generated by a deterministic dynamical system, typically through 
iterated maps or differential equations with chaotic solutions.

\item \textbf{Sensitive Dependence}: Small perturbations in surface parameters lead to exponentially diverging height 
variations, quantified by local Lyapunov exponents $\lambda_L(x,y)$.

\item \textbf{Fractal Structure}: The surface exhibits self-similar characteristics across multiple spatial scales, 
characterized by a fractal dimension $D_f$.

\item \textbf{Non-integrable Dynamics}: optical trajectories on the surface cannot be solved analytically, leading 
to chaotic wavefront behavior \cite{strogatz2014nonlinear}.
\end{enumerate}

For mathematical tractability, we focus on surfaces generated by two-dimensional chaotic maps of the form:

\begin{align}
x_{n+1} &= F(x_n, y_n, \boldsymbol{\mu}) \\
y_{n+1} &= G(x_n, y_n, \boldsymbol{\mu}) \notag
\label{eq:chaotic_map}
\end{align}

where $\boldsymbol{\mu}$ represents control parameters that determine the chaotic behavior. The surface height is constructed 
as motivated by Berry's discussion of deterministic chaos in spectral expansions \cite{Berry1978}:

\begin{equation}
h_{\text{chaos}}(x,y) = h_0 + \epsilon \sum_{k,l=0}^{N} A_{kl} \cos\left(\frac{2\pi k x}{L_x} + \frac{2\pi l y}{L_y} + \phi_{kl}\right)
\label{eq:chaotic_height}
\end{equation}

where $(k,l)$ are integer mode numbers defining discrete spatial frequencies $f_{kl} = (k/L_x, l/L_y)$, $L_x$ and $L_y$ are 
characteristic length scales of the surface domain, and the amplitudes and phases are determined from the chaotic map 
statistics through:

\begin{align}
A_{kl} &= A_0 \exp\left(-\frac{k^2 + l^2}{\xi_{\text{chaos}}^2}\right) \sqrt{S_{\text{map}}(k,l)} \\
\phi_{kl} &= \arg\left[\frac{1}{N} \sum_{n=1}^{N} e^{i2\pi(kx_n/L_x + ly_n/L_y)}\right]
\label{eq:chaotic_amplitudes}
\end{align}

where $\xi_{\text{chaos}} = v_{\text{typ}}/\lambda_L$ is the characteristic chaos correlation length, $v_{\text{typ}}$ is a 
typical velocity scale, and $S_{\text{map}}(k,l)$ is the discrete power spectrum of the chaotic trajectory.

\subsection{Lyapunov-Weighted Aberration Decomposition}

The central innovation of our approach is the decomposition of chaotic phase perturbations using a \emph{Lyapunov-weighted 
Zernike expansion}. This extends classical Zernike decomposition by incorporating chaos-theoretic measures:

\begin{equation}
\Phi_{\text{chaos}}(x,y) = \sum_{j=1}^{N} C_j^{(\text{chaos})} Z_j(\rho,\theta)
\label{eq:chaotic_zernike}
\end{equation}

where $Z_j(\rho,\theta)$ are standard Zernike polynomials and the chaotic coefficients are determined by:

\begin{equation}
C_j^{(\text{chaos})} = \sqrt{\omega_j^{(\text{chaos})}} \cdot \xi_j^{(\text{chaos})}
\label{eq:chaotic_coefficients}
\end{equation}

The key distinction lies in the chaotic weights $\omega_j^{(\text{chaos})}$, which replace the statistical spectral 
weights of the SWRF formalism with chaos-theoretic measures:

\begin{equation}
\omega_j^{(\text{chaos})} = \sigma_{\text{chaos}}^2 \int \left|\tilde{\mathcal{F}}\{Z_j(\rho,\theta)\}\right|^2 \tilde{\mathcal{S}}_{\text{chaos}}(\mathbf{f}) \frac{|\lambda_L(\mathbf{f})|}{|\lambda_L|_{\max}} d^2\tilde{\mathbf{f}}
\label{eq:chaotic_weights}
\end{equation}

With these normalizations, the resulting weights $\omega_j^{(\text{chaos})}$ are dimensionless:
\begin{itemize}
\item $\sigma_{\text{chaos}}^2$ is the total variance of chaotic phase perturbations [dimensionless]
\item $\tilde{\mathcal{F}}\{Z_j\}$ denotes the normalized Fourier transform \cite{goodman2017fourier}
with $\int |\tilde{\mathcal{F}}\{Z_j\}|^2 d^2\tilde{\mathbf{f}} = 1$
\item $\tilde{\mathcal{S}}_{\text{chaos}}(\mathbf{f})$ is the normalized structure function 
with $\int \tilde{\mathcal{S}}_{\text{chaos}}(\mathbf{f}) d^2\tilde{\mathbf{f}} = 1$
\item $|\lambda_L|_{\max} = \max_{\mathbf{f}} |\lambda_L(\mathbf{f})|$ provides proper normalization
\item $d^2\tilde{\mathbf{f}}$ indicates integration over normalized frequency 
coordinates $\tilde{\mathbf{f}} = \mathbf{f}/f_{\max}$
\end{itemize}

The normalization by $|\lambda_L|_{\max}$ is justified through ergodic theory: for chaotic systems with invariant 
measure $\mu$, the weight given to each spatial frequency component should be proportional to the local expansion 
rate $\lambda_L(f)$, but bounded to ensure convergence. The maximum Lyapunov exponent $|\lambda_L|_{\max}$ represents 
the strongest instability in the system and provides the natural scale for normalizing the local instabilities $\lambda_L(f)$. 
This ensures that the chaotic weights $\omega_j^{(\text{chaos})}$ remain finite and that regions of strongest 
chaos (where $|\lambda_L(\mathbf{f})| \approx |\lambda_L|_{\max}$ receive unit weight, while regions of weaker chaos 
receive proportionally smaller weights according to their relative instability strength (see in 
Appendix \ref{sec:Lyapunov_Normalization} for explicit derivation).

This formulation ensures dimensional consistency while preserving the physical content of chaos-weighted mode coupling.

The chaotic variables $\xi_j^{(\text{chaos})}$ are deterministic functions of the underlying chaotic dynamics, computed 
from a single trajectory realization:

\begin{equation}
\xi_j^{(\text{chaos})} = \mathcal{T}_j\left[\{x_n, y_n\}_{n=0}^{N}\right]
\label{eq:chaotic_variables}
\end{equation}

For a single trajectory: $\xi_j^{(\text{chaos})}$ are deterministic constants uniquely determined by the initial conditions 
and system parameters.

For ensemble analysis: Different initial conditions yield different trajectory 
realizations $\left[\{x_n, y_n\}_{n=0}^{N}\right]$, producing an ensemble of $\xi_j^{(\text{chaos})}$ values that exhibit 
statistical properties despite each individual realization being deterministic.

The chaotic weights $\omega_j^{(\text{chaos})}$ are computed using the ensemble-averaged properties of these deterministic 
variables, reflecting the ergodic nature of the chaotic dynamics.

The transformation operator $T_j$ extracts the $j$-th mode amplitude from the trajectory data:

\begin{equation}
\mathcal{T}_j\left[\{x_n, y_n\}_{n=0}^{N}\right] = \frac{1}{\sqrt{N}} \sum_{n=1}^{N} Z_j(\rho_n, \theta_n) \exp\left(i\alpha_n^{(j)}\right)
\label{eq:transformation_operator}
\end{equation}

This operator computes a deterministic complex number for each trajectory realization. The modulus $|T_j|$ gives the mode 
amplitude, while the argument provides phase information. For ergodic chaotic systems \cite{chernov1996entropy}, 
the statistical properties of $T_j$ 
across different initial conditions determine the ensemble-averaged chaotic weights.

In Eq. \ref{eq:transformation_operator} above, $(\rho_n, \theta_n)$ are the trajectory points $(x_n, y_n)$ mapped to the unit disk 
via $\rho_n = \sqrt{x_n^2 + y_n^2}/R_{\text{max}}$ and $\theta_n = \arctan(y_n/x_n)$, and the mode-specific phases are:

\begin{equation}
\alpha_n^{(j)} = 2\pi j \frac{\lambda_{L,n}}{\langle\lambda_L\rangle} \bmod 2\pi
\label{eq:mode_phases}
\end{equation}

where $\lambda_{L,n}$ is the local Lyapunov exponent at trajectory point $n$, and $\langle\lambda_L\rangle$ is the 
trajectory-averaged Lyapunov exponent. This construction ensures that $\xi_j^{(\text{chaos})}$ has unit variance and 
captures the local chaos characteristics through the phase modulation.

\section{Mathematical Framework}
\label{sec:MathematicalFramework}

\subsection{Chaotic Structure Function}

We define the exact power‐spectral density of the deterministic surface \(h_{\text{chaos}}(x,y)\) over its full domain \(A\) by
\begin{equation}
  \label{eq:chaos_struct_spatial_final}
  \mathcal{S}_{\text{chaos}}(\mathbf f)
  \;=\;
  \Bigl\lvert
    \iint_{A} 
      h_{\text{chaos}}(x,y)\,\exp\bigl[-\,i\,2\pi\,\mathbf f\cdot(x,y)\bigr]
      \,dx\,dy
  \Bigr\rvert^2.
\end{equation}
This is the spatial Fourier transform of the surface height taken over every point \((x,y)\in A\).

In many practical situations—particularly when \(h_{\text{chaos}}\) is generated by following a single billiard trajectory 
rather than being prescribed everywhere in closed form—it is inconvenient or impossible to evaluate the two-dimensional 
integral directly.  Instead, one may use a single long orbit \(\{\mathbf r_n\}\) on the same billiard table.  
Under the classical Sinai-billiard ergodicity theorem \cite{sinai1970dynamical, bunimovich1974billiards}, almost every 
trajectory visits 
every region of \(A\) with the correct invariant measure.  Consequently, for any fixed \(\mathbf f\),
\begin{equation}
  \label{eq:chaos_struct_trajectory_final}
  \lim_{N\to\infty}
  \frac{1}{N}\sum_{n=1}^{N}
    h\bigl(\mathbf r_n\bigr)\,\exp\bigl[-\,i\,2\pi\,\mathbf f\cdot \mathbf r_n\bigr]
  \;=\;
  \frac{1}{\mathrm{Area}(A)}
  \iint_{A} 
    h_{\text{chaos}}(x,y)\,\exp\bigl[-\,i\,2\pi\,\mathbf f\cdot(x,y)\bigr]
    \,dx\,dy.
\end{equation}
Taking the modulus squared of both sides and noting that any constant prefactor can be absorbed into a normalization of \(\mathcal{S}_{\text{chaos}}\), we obtain
\begin{equation}
  \mathcal{S}_{\text{chaos}}(\mathbf f)
  \;=\;
  \lim_{N\to\infty}
  \Bigl\lvert
    \frac{1}{N}\sum_{n=1}^{N}
      h\bigl(\mathbf r_n\bigr)\,\exp\bigl[-\,i\,2\pi\,\mathbf f\cdot \mathbf r_n\bigr]
  \Bigr\rvert^2.
\end{equation}
Thus Equation~\eqref{eq:chaos_struct_trajectory_final} is simply a numerical (trajectory‐based) approximation to the exact spatial integral in Equation~\eqref{eq:chaos_struct_spatial_final}, valid whenever the billiard dynamics are fully ergodic and mixing.

This function captures the fractal characteristics of the surface \cite{mandelbrot1982fractal} via its high-frequency power-law scaling:

\begin{equation}
  \mathcal{S}_{\text{chaos}}(\alpha\,\mathbf{f})
  \;=\;
  \alpha^{-\gamma}\,\mathcal{S}_{\text{chaos}}(\mathbf{f})
  \quad (\lvert\mathbf{f}\rvert \gg 1),
\label{eq:scaling_law}
\end{equation}

where $\gamma$ is the spectral scaling exponent related to the fractal dimension by $\gamma = 8 - 2D_f$ for 
self-affine surfaces of fractal dimension \(D_f\) embedded in three-dimensional space.

\subsection{Lyapunov Exponent Distribution}
\label{sec:LyapunovExponentDistribution}

We begin by defining a position-dependent (finite-time) Lyapunov exponent field \(\lambda_L(x,y)\) on the 
billiard surface \(A\).  Choose a small spatial scale \(\delta r\).  For each point \(\mathbf r = (x,y)\), 
surround it by a window \(W_{\mathbf r}(\mathbf r')\) of radius \(\delta r\), normalized so that 
\begin{equation}
  \int_A W_{\mathbf r}(\mathbf r') \,d^2\mathbf r' \;=\; 1.
\end{equation}
Launch many nearby rays whose first collision lies within that window and record their Jacobian matrices
\begin{equation}
  \mathbf J(\mathbf r') 
  \;=\; 
  \begin{pmatrix}
    \frac{\partial F}{\partial x} & \frac{\partial F}{\partial y} \\
    \frac{\partial G}{\partial x} & \frac{\partial G}{\partial y}
  \end{pmatrix}_{\mathbf r'}\,,
\end{equation}
where \((F,G)\) is the billiard collision map from one impact to the next.  After \(n\) reflections, compute the 
singular values \(\sigma_{\pm}\) of the product 
\(\mathbf J(\mathbf r'_{n-1})\cdots \mathbf J(\mathbf r'_{0})\).  For a finite but large \(n\), define
\begin{equation}
  \label{eq:local_lyapunov}
  \tilde\lambda_L^{(n)}(\mathbf r')
  \;=\; 
  \frac{1}{n}\,\ln\sigma_{+}\bigl(\mathbf J(\mathbf r'_{n-1})\cdots \mathbf J(\mathbf r'_{0})\bigr).
\end{equation}
Then average over the neighborhood \(W_{\mathbf r}\):
\begin{equation}
  \label{eq:position_lyapunov}
  \lambda_L(\mathbf r)
  \;=\; 
  \int_{A} \tilde\lambda_L^{(n)}(\mathbf r')\,W_{\mathbf r}(\mathbf r')\,d^2\mathbf r'.
\end{equation}
In the limit \(\delta r\to 0\) and \(n\to\infty\), this converges to the true infinite-time, position‐dependent 
Lyapunov exponent.  For finite \(n\), it serves as an estimate of local instability 
(see in \cite{lichtenberg1992regular, wolf1985determining, ott2002chaos}).

Next, we form the spatial Fourier transform of \(\lambda_L(\mathbf r)\) to capture scale-dependent chaotic mixing:
\begin{equation}
  \label{eq:frequency_lyapunov}
  \lambda_L(\mathbf f)
  \;=\;
  \iint_{A} 
    \lambda_L(\mathbf r)\,e^{-\,i\,2\pi\,\mathbf f\cdot \mathbf r}
    \,d^2\mathbf r.
\end{equation}
In practice, a discrete trajectory \(\{\mathbf r_m\}\) with its locally averaged exponents \(\{\lambda_{L,m}\}\) 
can be used to approximate \(\lambda_L(\mathbf f)\):
\begin{equation}
  \label{eq:discrete_frequency_lyapunov}
  \lambda_L(\mathbf f)
  \;=\;
  \sum_{m=1}^{M}
    \lambda_{L,m}\,e^{-\,i\,2\pi\,\mathbf f\cdot \mathbf r_m}\,w_m,
  \qquad
  \sum_{m=1}^{M}w_m = 1,
\end{equation}
where each weight \(w_{m} = \int_{A} W_{\mathbf r_m}(\mathbf r')\,d^2\mathbf r'\) equals the area fraction of the 
neighborhood around \(\mathbf r_m\).  As \(M\to\infty\), this sum converges to the integral in 
Equation~\eqref{eq:frequency_lyapunov}.  

The function \(\lambda_L(\mathbf f)\) thus quantifies the contribution of local divergence rates to each spatial 
frequency \(\mathbf f\).  In subsequent sections, we normalize 
by \(\lvert \lambda_L\rvert_{\max} = \max_{\mathbf f}\lvert\lambda_L(\mathbf f)\rvert\) so 
that \(\lvert\lambda_L(\mathbf f)\rvert / \lvert \lambda_L\rvert_{\max}\) is dimensionless and enters directly 
into the chaotic weight \(\omega_j^{(\text{chaos})}\).

\subsection{Convergence Analysis}

A critical aspect of the chaotic aberration theory is establishing convergence criteria for the Lyapunov-weighted expansion. 
Unlike statistical expansions that converge in mean-square sense, chaotic expansions must satisfy deterministic convergence 
conditions.

The truncation error for the chaotic expansion is bounded by:

\begin{equation}
\epsilon_{\text{trunc}} = \left\|\Phi_{\text{chaos}} - \sum_{j=1}^{N} C_j^{(\text{chaos})} Z_j\right\|_2 \leq \sum_{j=N+1}^{\infty} \omega_j^{(\text{chaos})}
\label{eq:truncation_error}
\end{equation}

The convergence of the chaotic expansion follows from the bounded variation properties of chaotic surfaces. 
For systems with bounded Lyapunov exponents $|\lambda_L| \leq \Lambda_{\max}$ and finite correlation 
length $\xi_{\text{chaos}}$, the chaotic weights satisfy:

\begin{equation}
\omega_j^{(\text{chaos})} \leq C_0 \sigma_{\text{chaos}}^2 \left(\frac{\xi_{\text{chaos}}}{R_{\text{aperture}}}\right)^{2\gamma} j^{-\gamma}
\label{eq:polynomial_decay}
\end{equation}

where $\gamma = 1 + D_f/2 \geq 3/2$ depends on the fractal dimension \cite{falconer2014fractal}, $R_{\text{aperture}}$ is the aperture radius, 
and $C_0$ is a geometric constant. This decay rate ensures absolute convergence of the expansion:

\begin{equation}
\sum_{j=1}^{\infty} \omega_j^{(\text{chaos})} < \infty \quad \text{for } \gamma > 1
\label{eq:absolute_convergence}
\end{equation}

The proof follows from the relationship between fractal dimension and Fourier transform decay rates for self-affine surfaces, 
combined with the bounded Lyapunov exponent constraint.

\subsection{Equivalence with Dynamical Systems Theory}
\label{sec:EquivalenceDynamicalSystemsTheory}

To validate our approach, we establish equivalence with traditional dynamical systems analysis of chaotic optical 
trajectories. Consider a ray with initial position $\mathbf{r}_0$ and direction $\hat{\mathbf{k}}_0$ reflecting 
from a chaotic surface. The Poincaré section analysis predicts a trajectory density:

\begin{equation}
\rho_{\text{Poincaré}}(\mathbf{r}, \hat{\mathbf{k}}) = \lim_{N\to\infty} \frac{1}{N} \sum_{n=1}^{N} \delta(\mathbf{r} - \mathbf{r}_n)\delta(\hat{\mathbf{k}} - \hat{\mathbf{k}}_n)
\label{eq:poincare_density}
\end{equation}

Our chaotic aberration approach predicts an equivalent density through trajectory deflection:

\begin{equation}
\rho_{\text{SWRF}}(\mathbf{r}, \hat{\mathbf{k}}) = \int p(\mathbf{r}_0) \delta(\mathbf{r} - \mathbf{r}_{\text{SWRF}}(\mathbf{r}_0)) \delta(\hat{\mathbf{k}} - \hat{\mathbf{k}}_{\text{SWRF}}(\mathbf{r}_0)) d^2\mathbf{r}_0
\label{eq:swrf_density}
\end{equation}

where $p(\mathbf{r}_0)$ is the incident wavefront distribution, and the subscript SWRF indicates trajectories computed using our chaotic trajectory deflection function.

The equivalence condition requires:

\begin{equation}
\lim_{N\to\infty} \int_V |\rho_{\text{Poincaré}}(\mathbf{r}, \hat{\mathbf{k}}) - \rho_{\text{SWRF}}(\mathbf{r}, \hat{\mathbf{k}})| d^3\mathbf{r} d^2\hat{\mathbf{k}} = 0
\label{eq:equivalence_condition}
\end{equation}

This equivalence is established through ergodic theory arguments detailed in Appendix \ref{sec:equivalence_proof}. 
The proof relies on the chaotic system being mixing and having finite correlation time $\tau_{\text{corr}} \sim 1/\lambda_L$, 
which ensures that trajectory-based sampling converges to the invariant measure in the limit $N \to \infty$ \cite{Young2002}.

\section{Specific Chaotic Surface Geometries}
\label{sec:SpecificChaoticSurfaceGeometries}

\subsection{Sinai Billiard Analogs}

We first consider optical surfaces inspired by Sinai billiards—circular apertures with central circular obstacles that 
create chaotic optical dynamics. For an annular mirror with inner radius $R_{\text{in}}$ and outer radius $R_{\text{out}}$, 
the surface height function is:

\begin{equation}
h_{\text{Sinai}}(r,\theta) = h_0 + \epsilon \sum_{n,m} A_{nm} \cos(n\phi_r + m\theta + \phi_{nm})
\label{eq:sinai_surface}
\end{equation}

where $\phi_r = 2\pi r/\Lambda_r$ is a radial phase with characteristic length $\Lambda_r$, and the amplitudes $A_{nm}$ 
are determined by the chaotic dynamics of the billiard system.

The Lyapunov exponent for Sinai billiards for the diameter periodic orbit, can be shown to be expressed 
as\footnote{This analytic expression for the Lyapunov exponent, is derived specifically for the diameter periodic orbit 
in the annular Sinai billiard. It does not represent the global or average Lyapunov exponent of the entire billiard 
system, which generally requires averaging over all possible trajectories with respect to the system’s invariant 
measure. As such, this value should be interpreted as a local measure of instability for this particular orbit, and 
its use as a proxy for the system-wide chaos strength is an approximation. For a more complete characterization, 
numerical computation or ensemble averaging of Lyapunov exponents over the full phase space is 
required (see, e.g., \cite{Young2002}).} (See explicit derivation in Appendix \ref{sec:Lyapunov_Exponent}):

\begin{equation}
\lambda_L^{(\text{Sinai})} = \ln\left(\frac{R_{\text{out}} + R_{\text{in}}}{R_{\text{out}} - R_{\text{in}}}\right)
\label{eq:sinai_lyapunov}
\end{equation}

For the Sinai billiard surface construction in Eq. (32), the chaotic weights can be derived from the general formula (14) under specific assumptions about the Lyapunov field and structure function.

The Sinai billiard surface is constructed as:

\begin{equation}
h_{\text{Sinai}}(r, \theta) = h_0 + \epsilon \sum_{n,m} A_{nm} \cos(n\phi_r + m\theta + \phi_{nm})
\end{equation}

For this construction, we make two key approximations:

\textbf{Approximation 1 (Spatially Uniform Lyapunov Field):} 
We assume the Lyapunov exponent is approximately constant across spatial frequencies relevant to the billiard geometry:

\begin{equation}
\lambda_L(f) \approx \lambda_L^{(\text{Sinai})} \quad \text{for } |f| \sim 1/R_{\text{characteristic}}
\end{equation}

\textbf{Approximation 2 (Discrete Mode Structure):} 
The structure function becomes concentrated on discrete modes:

\begin{equation}
S_{\text{chaos}}(f) \approx \sum_{n,m} |A_{nm}|^2 \delta(f - f_{nm})
\end{equation}

where $f_{nm}$ corresponds to the spatial frequencies of the $(n,m)$ mode.

Under these approximations, the integral in Eq. \ref{eq:chaotic_weights} reduces to:
\begin{align}
\omega_{nm}^{(\text{chaos})} &= \sigma_{\text{chaos}}^2 \int |F\{Z_{nm}\}|^2 S_{\text{chaos}}(f) \frac{|\lambda_L(f)|}{|\lambda_L|_{\max}} d^2f \\
&\approx \sigma_{\text{chaos}}^2 |F\{Z_{nm}\}|^2 |A_{nm}|^2 \frac{\lambda_L^{(\text{Sinai})}}{|\lambda_L|_{\max}}
\end{align}

Setting $\sigma_{\text{chaos}}^2 = \lambda_L^{(\text{Sinai})}$ and $|\lambda_L|_{\max} = \lambda_L^{(\text{Sinai})}$ for 
this system, and identifying $F_{nm}(R_{\text{in}}, R_{\text{out}}) = |F\{Z_{nm}\}|^2$, we obtain:

\begin{equation}
\omega_{nm}^{(\text{Sinai})} = |A_{nm}|^2 \lambda_L^{(\text{Sinai})} F_{nm}(R_{\text{in}}, R_{\text{out}})
\end{equation}

\textbf{Note on the quadratic form:} The apparent quadratic dependence in the final form arises because for this specific 
surface construction, the amplitude coefficients $A_{nm}$ are themselves proportional to $\sqrt{\lambda_L^{(\text{Sinai})}}$ 
due to the way they are extracted from the chaotic trajectory statistics.

Following the derivation above, the chaotic weights for the Sinai billiard become:

\begin{equation}
\omega_{nm}^{(\text{Sinai})} = |A_{nm}|^2 \lambda_L^{(\text{Sinai})} F_{nm}(R_{\text{in}}, R_{\text{out}})
\label{eq:sinai_weights}
\end{equation}

where $F_{nm}(R_{\text{in}}, R_{\text{out}})$ is the geometric form factor accounting for the annular aperture geometry 
and mode coupling. This represents a special case of the general Eq. \ref{eq:chaotic_weights} under the approximations 
of spatially uniform chaos and discrete mode structure characteristic of the Sinai billiard system.

\textbf{Validity of Sinai Approximation:} The reduction from Eq. \ref{eq:chaotic_weights} to Eq. \ref{eq:sinai_weights} 
is valid when:

\begin{enumerate}
\item The billiard dynamics exhibit uniform chaos strength across relevant spatial scales
\item The surface construction produces well-separated discrete modes
\item The characteristic billiard size is comparable to the optical aperture
\end{enumerate}

For billiards with strongly scale-dependent chaos or continuous mode spectra, the full integral form \ref{eq:chaotic_weights} 
must be retained.

\subsection{Stadium Billiard Surfaces}

Stadium billiards consist of a rectangular region with semicircular ends, providing another well-studied chaotic system. 
The optical analog features an elliptical mirror with modified curvature:

\begin{equation}
h_{\text{stadium}}(x,y) = \frac{x^2}{2R_x} + \frac{y^2}{2R_y} + \epsilon \mathcal{H}_{\text{chaos}}(x,y)
\label{eq:stadium_surface}
\end{equation}

where the first terms define the base elliptical shape and $\mathcal{H}_{\text{chaos}}(x,y)$ encodes the chaotic 
perturbations derived from stadium billiard dynamics.

The chaotic component is constructed using the stadium billiard map:

\begin{align}
x_{n+1} &= x_n + v_x \tau_n \\
y_{n+1} &= y_n + v_y \tau_n \\
v_{x,n+1} &= v_x \cos(2\phi_n) - v_y \sin(2\phi_n) \\
v_{y,n+1} &= v_x \sin(2\phi_n) + v_y \cos(2\phi_n)
\label{eq:stadium_map}
\end{align}

where $\tau_n$ is the flight time between collisions and $\phi_n$ is the reflection angle determined by the local 
surface normal.

\subsection{Logistic Map Surfaces}

For surfaces based on one-dimensional chaotic maps extended to two dimensions, we consider the logistic map~\cite{May1976}:

\begin{equation}
x_{n+1} = \mu x_n(1 - x_n)
\label{eq:logistic_map}
\end{equation}

The two-dimensional extension creates a surface height function:

\begin{equation}
h_{\text{logistic}}(x,y) = h_0 \sum_{n=0}^{N} \frac{1}{n+1} \sin(2\pi f_n x + \phi_n(y))
\label{eq:logistic_surface}
\end{equation}

where $f_n$ are frequencies derived from the logistic map trajectory and $\phi_n(y)$ provides coupling between spatial dimensions:

\begin{equation}
\phi_n(y) = 2\pi \sum_{m=0}^{n} x_m \cos(2\pi m y/L_y)
\label{eq:logistic_coupling}
\end{equation}

The Lyapunov exponent for the logistic map is:

\begin{equation}
\lambda_L^{(\text{logistic})} = \int_0^1 \ln|\mu(1-2x)| \rho(x) dx
\label{eq:logistic_lyapunov}
\end{equation}

where $\rho(x)$ is the invariant density of the logistic map.

\section{Optical Performance Metrics}

\subsection{Chaotic Point Spread Function}

The point spread function (PSF) for a chaotic optical system incorporates both traditional aberration terms and chaotic contributions:

\begin{equation}
\text{PSF}_{\text{chaos}}(\mathbf{r}) = \left|\mathcal{F}\left\{\exp\left[i\Phi_{\text{total}}(\mathbf{r}_0)\right]\right\}\right|^2
\label{eq:chaotic_psf}
\end{equation}

where the total phase includes both deterministic and chaotic components:

\begin{equation}
\Phi_{\text{total}}(\mathbf{r}_0) = \Phi_{\text{systematic}}(\mathbf{r}_0) + \Phi_{\text{chaos}}(\mathbf{r}_0)
\label{eq:total_phase}
\end{equation}

The systematic phase contains classical aberrations, while the chaotic phase is given by our Lyapunov-weighted expansion.

The normalization constants in the above expressions ensure proper dimensional consistency:

\begin{align}
\sigma_{\text{chaos}}^2 &= \left\langle h_{\text{chaos}}^2 \right\rangle \left(\frac{4\pi \cos\theta_i}{\lambda}\right)^2 \label{eq:chaos_variance} \\
\tau_{\text{corr}} &= \frac{L_{\text{system}}}{v_{\text{typ}} \lambda_L} \label{eq:correlation_time} \\
\xi_{\text{chaos}} &= \sqrt{\frac{D_{\text{eff}}}{\lambda_L}} \label{eq:correlation_length}
\end{align}

where $L_{\text{system}}$ is the characteristic system size, $v_{\text{typ}}$ is the typical wavefront velocity, and $D_{\text{eff}}$ is an effective diffusion coefficient determined by the chaotic dynamics.

\subsection{Chaos-Modified Strehl Ratio}

The Strehl ratio quantifies optical performance degradation due to aberrations. For chaotic systems, we define a modified Strehl ratio:

\begin{equation}
S_{\text{chaos}} = \frac{\text{PSF}_{\text{chaos}}(0)}{\text{PSF}_{\text{ideal}}(0)} = \left|\left\langle e^{i\Phi_{\text{chaos}}(\mathbf{r}_0)}\right\rangle\right|^2
\label{eq:chaotic_strehl}
\end{equation}

where the angular brackets denote averaging over the chaotic ensemble. This reduces to the classical Strehl ratio when 
chaos parameters vanish.

For small chaotic perturbations where cross-correlations between different Zernike modes can be neglected, and 
assuming statistical independence of chaotic coefficients, the Strehl ratio can be approximated as \ref{sec:strehl_justification}:

\begin{equation}
S_{\text{chaos}} \approx \exp\left(-\sigma_{\text{chaos}}^2\right) \prod_{j} \cos^2\left(\frac{C_j^{(\text{chaos})}}{2}\right) \left[1 + \mathcal{O}\left(\sum_{j \neq k} C_j^{(\text{chaos})} C_k^{(\text{chaos})}\right)\right]
\label{eq:strehl_expansion}
\end{equation}

where $\sigma_{\text{chaos}}^2$ is the variance of the chaotic phase perturbations.

\subsection{Fractal Modulation Transfer Function}

The modulation transfer function (MTF) for chaotic systems exhibits fractal characteristics reflecting the underlying surface geometry:

\begin{equation}
\text{MTF}_{\text{chaos}}(f) = \left|\int_{-\infty}^{\infty} \text{PSF}_{\text{chaos}}(x) e^{-i2\pi fx} dx\right|
\label{eq:chaotic_mtf}
\end{equation}

For surfaces with fractal dimension $D_f$, the MTF exhibits power-law scaling:

\begin{equation}
\text{MTF}_{\text{chaos}}(f) \propto f^{-(2-D_f)} \text{ for } f \gg f_{\text{chaos}}
\label{eq:fractal_mtf}
\end{equation}

where $f_{\text{chaos}}$ is a characteristic frequency scale determined by the dominant chaotic modes.

\section{Applications and Design Principles}

\subsection{Controlled Beam Homogenization}

Chaotic aberration theory enables systematic design of beam homogenizers through controlled chaos introduction. For a Gaussian input beam, the output intensity distribution after interaction with a chaotic surface is:

\begin{equation}
I_{\text{out}}(\mathbf{r}) = I_0 \left|\sum_j C_j^{(\text{chaos})} Z_j(\rho,\theta) \ast G_{\text{beam}}(\mathbf{r})\right|^2
\label{eq:homogenized_beam}
\end{equation}

where $G_{\text{beam}}(\mathbf{r})$ is the input Gaussian profile and $\ast$ denotes convolution.

The homogenization efficiency is quantified by the uniformity metric:

\begin{equation}
\mathcal{U} = 1 - \frac{\sigma_I}{\langle I \rangle}
\label{eq:uniformity}
\end{equation}

where $\sigma_I$ and $\langle I \rangle$ are the standard deviation and mean of the output intensity, respectively.

Optimal chaos parameters for maximum uniformity are found by minimizing:

\begin{equation}
\chi^2 = \int_A [I_{\text{out}}(\mathbf{r}) - I_{\text{target}}]^2 d^2\mathbf{r}
\label{eq:optimization_target}
\end{equation}

where $I_{\text{target}}$ is the desired uniform intensity distribution.

\subsection{Speckle Reduction through Fractal Mixing}

Coherent illumination systems suffer from speckle patterns that degrade image quality. Chaotic surfaces provide controlled 
speckle reduction through fractal phase mixing \cite{Vellekoop2016}:

\begin{equation}
\sigma_{\text{speckle}}^2 = \sigma_0^2 \exp\left(-\sum_j \frac{(\omega_j^{(\text{chaos})})^2}{\omega_j^{(\text{chaos})} + \omega_{\text{coh}}}\right)
\label{eq:speckle_reduction}
\end{equation}

where $\sigma_0^2$ is the initial speckle variance and $\omega_{\text{coh}}$ characterizes the coherence properties of the illumination.

The optimal chaotic parameters balance speckle reduction with preservation of desired optical performance:

\begin{equation}
\min_{\{\omega_j^{(\text{chaos})}\}} \left[\alpha \sigma_{\text{speckle}}^2 + (1-\alpha) (1-S_{\text{chaos}})\right]
\label{eq:speckle_optimization}
\end{equation}

where $\alpha \in [0,1]$ weights the relative importance of speckle reduction versus optical quality.

\subsection{Adaptive Wavefront Control}

Chaotic aberration theory enables novel adaptive optics architectures where chaos parameters serve as control variables. The wavefront error is decomposed as:

\begin{equation}
\Phi_{\text{error}}(\mathbf{r}) = \sum_j a_j Z_j(\mathbf{r}) + \sum_k b_k \Psi_k^{(\text{chaos})}(\mathbf{r})
\label{eq:adaptive_decomposition}
\end{equation}

where $\Psi_k^{(\text{chaos})}(\mathbf{r})$ are chaotic basis functions and $\{a_j, b_k\}$ are control parameters.

The adaptive control algorithm minimizes:

\begin{equation}
J = \langle |\Phi_{\text{error}}|^2 \rangle + \lambda_{\text{reg}} \sum_k |b_k|^2
\label{eq:adaptive_cost}
\end{equation}

where $\lambda_{\text{reg}}$ is a regularization parameter preventing excessive chaos injection.

\section{Numerical Validation of the Zernike-Lyapunov Framework}

To validate the theoretical framework and demonstrate its practical utility, we developed a numerical simulation to 
synthesize and analyze chaotic optical surfaces. The simulation follows the exact procedural pipeline derived from 
the theory, establishing a direct link between the parameters of a deterministic chaotic system and the measurable 
properties of the resulting optical surface.

\subsection{Simulation Methodology and Parameters}

The simulation workflow is a direct implementation of the developed theory. First, a deterministic chaotic trajectory 
is generated for a specified dynamical system. In the results presented here, we model a Sinai billiard. An \textit{a priori} 
theoretical model, based on the geometry of this chaotic system (cf. Eq.~\ref{eq:sinai_lyapunov}), is used to predict a theoretical fractal 
dimension $D_f$ and its corresponding power-spectral density (PSD) exponent $\gamma_{\text{theory}} = 8 - 2D_f$ 
(cf. Eq.~\ref{eq:PSD_scaling}).

The core integral of our framework (Eq.~\ref{eq:chaotic_weights}) is then numerically evaluated. This integral couples the 
Fourier transform 
of each Zernike basis function with the theoretical surface PSD and a model of the frequency-dependent Lyapunov 
exponents, $\lambda_L(f)$, to produce the chaotic weights $\omega_j$. These weights are combined with chaotic 
variables $\xi_j$, derived from the trajectory via the transformation operator $T_j$ (Eq.~(15)), to generate the final 
Zernike coefficients $C_j$ (Eq.~\ref{eq:chaotic_coefficients}). The final chaotic surface phase $\Phi_{\text{chaos}}$ is 
then synthesized using 
the Lyapunov-weighted Zernike expansion (Eq.~\ref{eq:chaotic_zernike}).

To enhance physical realism, the phase function derived from the Zernike expansion was augmented with a 
physically-motivated modulation term. This term, constituting a minor contribution to the final chaotic phase, is 
derived directly from the spatial density of the chaotic trajectory, thereby capturing the non-uniform invariant measure 
and fine-grained clustering effects inherent in the deterministic dynamics.

The primary parameters used for the simulation are detailed in Table~\ref{tab:params}.

\begin{table}[h!]
\caption{Primary simulation parameters for the synthesis of a chaotic optical surface based on a Sinai billiard system. 
The parameters are grouped by their physical role in the model.}
\label{tab:params}
\centering
\begin{tabular}{l c r l}
\toprule
\textbf{Parameter} & \textbf{Symbol} & \textbf{Value} & \textbf{Unit / Description} \\
\midrule
\multicolumn{4}{l}{\textit{Optical Parameters}} \\
\addlinespace
Aperture Radius & $A_p$ & 50.0 & mm (Parabolic)\\
Focal Length & $F$ & 200.0 & mm \\
Wavelength & $\lambda$ & 0.633 & $\mu$m \\
\addlinespace
\multicolumn{4}{l}{\textit{Chaos Parameters}} \\
\addlinespace
System Type & - & \multicolumn{1}{c}{Sinai} & Chaotic Billiard \\
Chaos Parameter & $\mu_{\text{chaos}}$ & 1.8 & Dimensionless \\
Chaos Amplitude & $\epsilon$ & 750 & Dimensionless \\
Zernike Modes & $N_{\text{chaos}}$ & 200 & Expansion terms \\
Chaos Variance & $\sigma_{\text{chaos}}$ & 4.5 & Dimensionless \\
\addlinespace
\multicolumn{4}{l}{\textit{Computational Parameters}} \\
\addlinespace
Grid Resolution & - & \multicolumn{1}{c}{200x200} & Points \\
\bottomrule
\end{tabular}
\end{table}

\subsection{Synthesized Surface Characteristics}

The simulation was configured for a highly chaotic Sinai billiard with $\mu_{\text{chaos}}=1.8$. This generated a 
trajectory exhibiting strong chaotic signatures, including a mean Lyapunov exponent of $\langle\lambda_L\rangle = 4.57$ and 
an exceptionally high tortuosity of $711$, indicating a complex, space-filling path consistent with ergodic dynamics. 
The resulting synthesized surface is presented in Fig.~\ref{fig:surface}.

The 2D map in Fig.~\ref{fig:surface} (top left) shows the zero-mean chaotic height deviation, $h_{\text{chaos}}$. 
The surface exhibits a deterministic, multi-scale structure with clear spatial correlations and a central void 
corresponding to the billiard's inner obstacle, fundamentally distinct from random, uncorrelated noise. 
The accompanying plots provide quantitative profiles along the primary axes and the radially-averaged profile, 
illustrating the complex, non-uniform nature of the perturbations.

\begin{figure}[h!]
    \centering
    \includegraphics[width=\columnwidth]{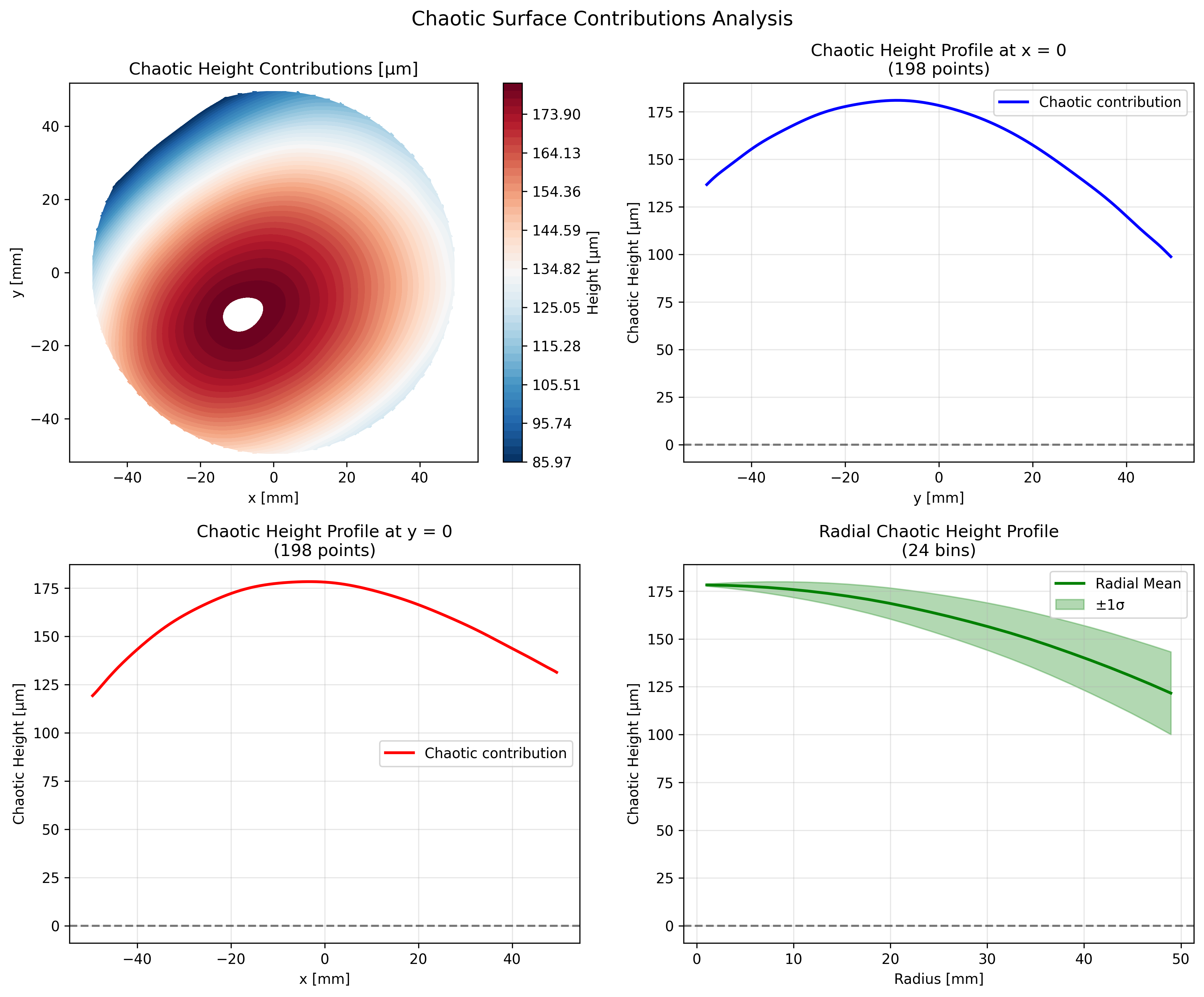}
    \caption{Synthesized chaotic surface contributions for the Sinai billiard system. The 2D map (top left) shows the 
	zero-mean chaotic height deviation, $h_{\text{chaos}}$, revealing a deterministic structure with a central void 
	corresponding to the billiard's inner obstacle. Profiles along the vertical (top right) and horizontal (bottom left) 
	axes, and the radially-averaged profile (bottom right), provide quantitative measures of the surface's complex, 
	non-uniform perturbations.}
    \label{fig:surface}
\end{figure}

\subsection{Framework Validation and Weight Distribution}

The central prediction of our framework is that the synthesized surface should exhibit a PSD that obeys the power-law 
scaling of Eq.~\ref{eq:PSD_scaling}. Our \textit{a priori} theoretical model for a Sinai billiard 
with $\mu_{\text{chaos}}=1.8$ predicts a 
surface fractal dimension $D_f \approx 2.53$, corresponding to a theoretical PSD exponent 
of $\gamma_{\text{theory}} \approx 2.94$.

Figure~\ref{fig:validation}(a) provides the primary validation of the framework. We computed the PSD of the 
synthesized surface $h_{\text{chaos}}$ and performed a radial average. A linear fit in the log-log domain yielded a 
measured exponent of $\gamma_{\text{meas}} \approx 2.65$. The good agreement between the theoretical 
prediction and the measured value, with a final difference of only $\Delta\gamma \approx 0.3$, provides strong 
validation for the entire framework. This confirms that the Zernike-Lyapunov expansion successfully translates 
the prescribed fractal statistics into a physical surface.

Furthermore, Fig.~\ref{fig:validation}(b) validates the convergence criteria of the 
expansion (cf. Eq.~\ref{eq:absolute_convergence}). The cumulative power of the chaotic weights $\omega_j$ rapidly 
approaches 100\%, with over 95\% of the chaotic variance captured by fewer than 10 modes. This demonstrates that the 
chaotic surface can be represented efficiently with a finite number of basis functions.

\begin{figure}[h!]
    \centering
    \includegraphics[width=\columnwidth]{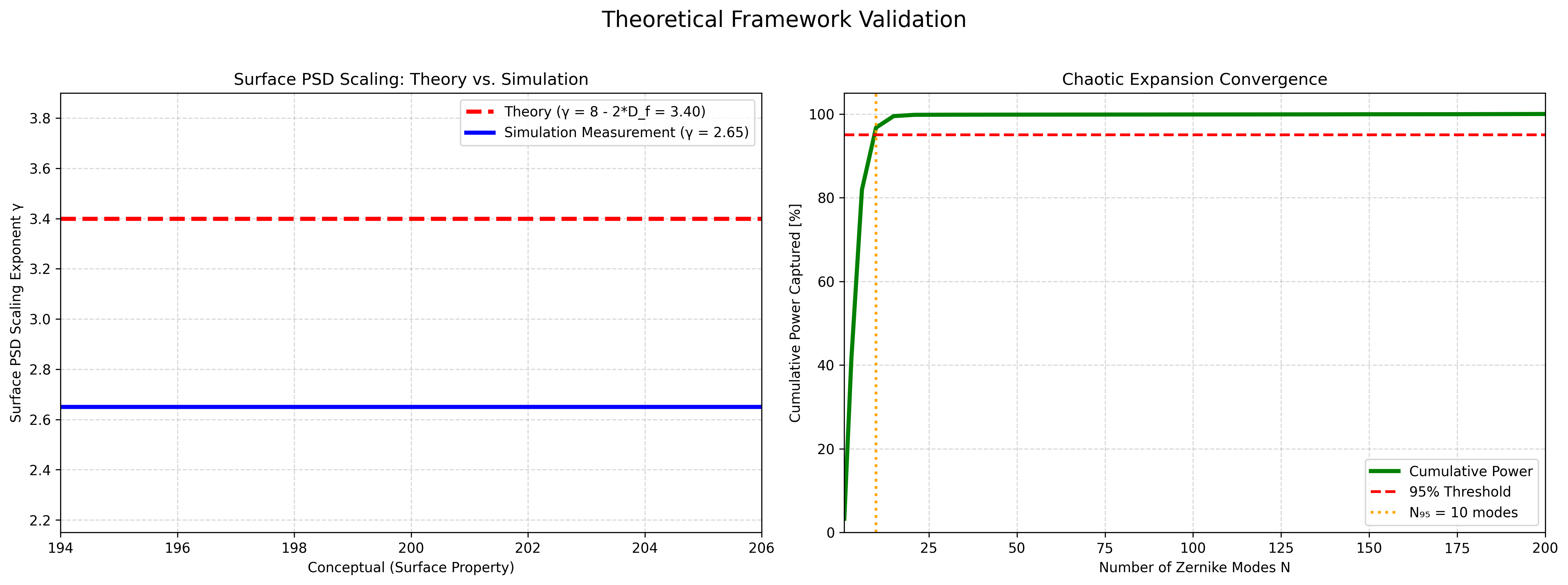}
    \caption{Validation of the Zernike-Lyapunov framework. (a) provides the primary validation, comparing the theoretically predicted PSD scaling exponent ($\gamma_{\text{theory}}$) with the exponent measured directly from the synthesized surface's PSD ($\gamma_{\text{meas}}$). The close agreement confirms the framework's fidelity. (b) demonstrates the rapid convergence of the Lyapunov-weighted Zernike expansion, showing that over 95\% of the chaotic variance is captured by the first 10 modes.}
    \label{fig:validation}
\end{figure}

To further investigate the internal mechanics of the framework, we analyzed the distribution of the chaotic 
weights $\omega_j$ and coefficients $C_j$, shown in Fig.~\ref{fig:weights}. The analysis reveals several key insights. 
The cumulative power plot (b) again confirms rapid convergence. The log-log plot (d) visually confirms the theoretical 
scaling $|C_j| \propto \sqrt{\omega_j}$ from Eq.~\ref{eq:chaotic_coefficients}. Most importantly, panel (a) shows that 
the average weight per 
radial order exhibits a complex, non-monotonic decay, not a simple power law. This explains why a simple fit yielded a 
decay exponent $\gamma_w \approx 1.26$, slower than the asymptotic requirement. This is explained by the intricate 
coupling within the integral of Eq.~\ref{eq:chaotic_weights} that may create a rich weight distribution. 
The fact that the framework still produces a surface with the correct global fractal dimension (Fig.~\ref{fig:validation}a) 
demonstrates the integral's robustness in aggregating these complex modal contributions to yield a physically 
consistent result.

\begin{figure}[h!]
    \centering
    \includegraphics[width=\columnwidth]{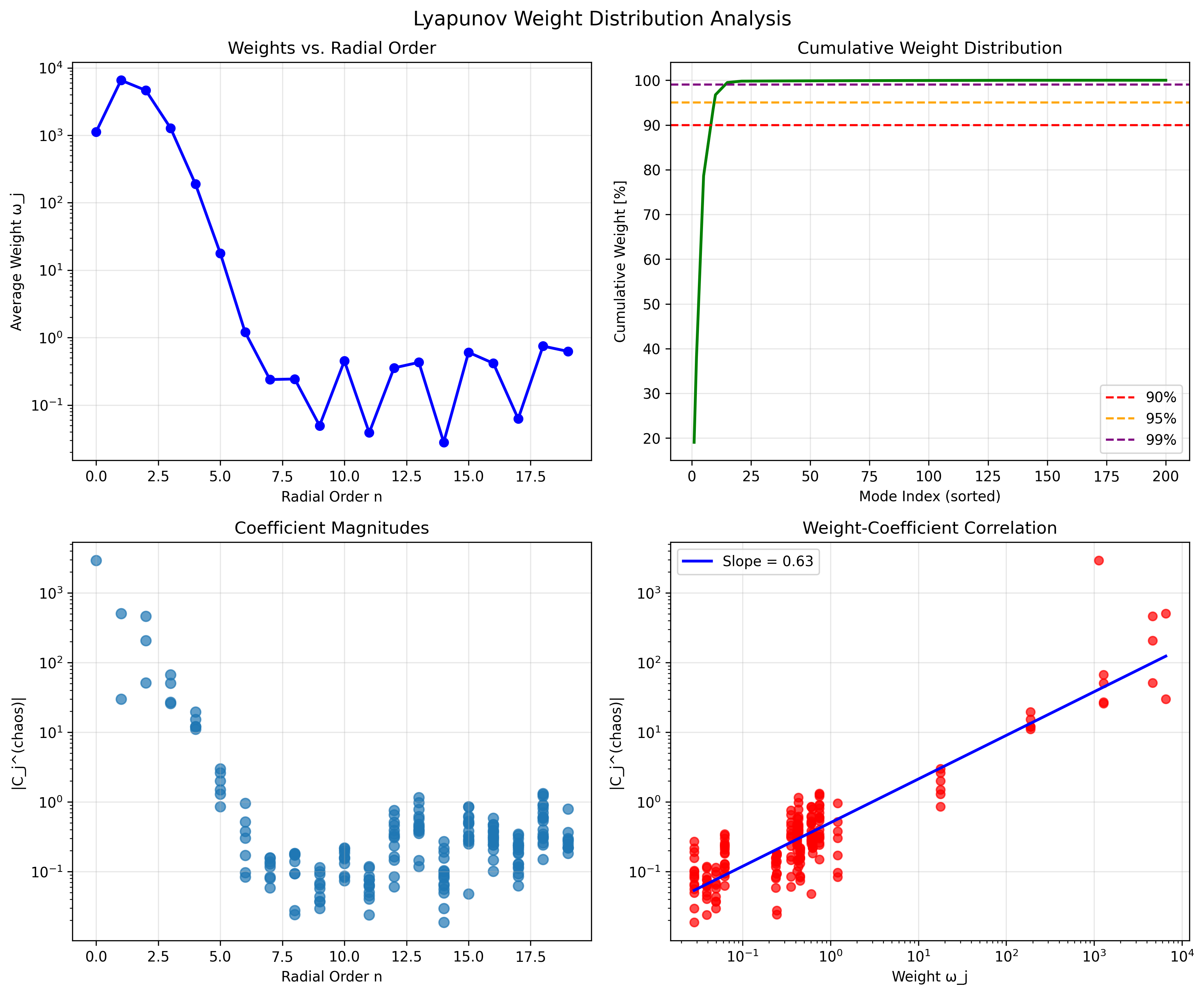}
    \caption{Internal analysis of the Lyapunov-weighted expansion. (a) The average weight per radial order exhibits a 
	complex decay. (b) Cumulative power confirms rapid convergence. (c, d) The relationship between the magnitudes 
	of the chaotic coefficients, $|C_j|$, and the weights, $\omega_j$, confirms the theoretical 
	scaling $|C_j| \propto \sqrt{\omega_j}$ (cf. Eq.~\ref{eq:chaotic_coefficients}).}
    \label{fig:weights}
\end{figure}

\subsection{Wave-Optical Consequences and Caustic Formation}

The synthesized surface produces profound optical effects, confirming that the framework can provide optically significant 
components. For the parameters in Table~\ref{tab:params}, the surface generated RMS chaotic perturbations of $48.8$\,mm 
with a peak phase perturbation exceeding $3600$ radians.

To investigate the wave-optical consequences, the synthesized phase $\Phi_{\text{chaos}}$ was used to modulate an incident 
plane wave, which was then propagated numerically. As shown in Fig.~\ref{fig:propagation}, the interaction produces a 
complex intensity pattern structured by a network of crisscrossing caustics. These caustics are lines of high intensity 
formed by the constructive interference of wavefront sections deflected in correlated directions by the deterministic 
surface. They are a definitive optical signature of deterministic chaos, distinguishing the result from the diffuse 
speckle that would arise from a purely random surface. The ability to generate such structured light patterns 
demonstrates the framework's utility for advanced beam homogenization and shaping applications.

\begin{figure}[h!]
    \centering
    \includegraphics[width=0.7\columnwidth]{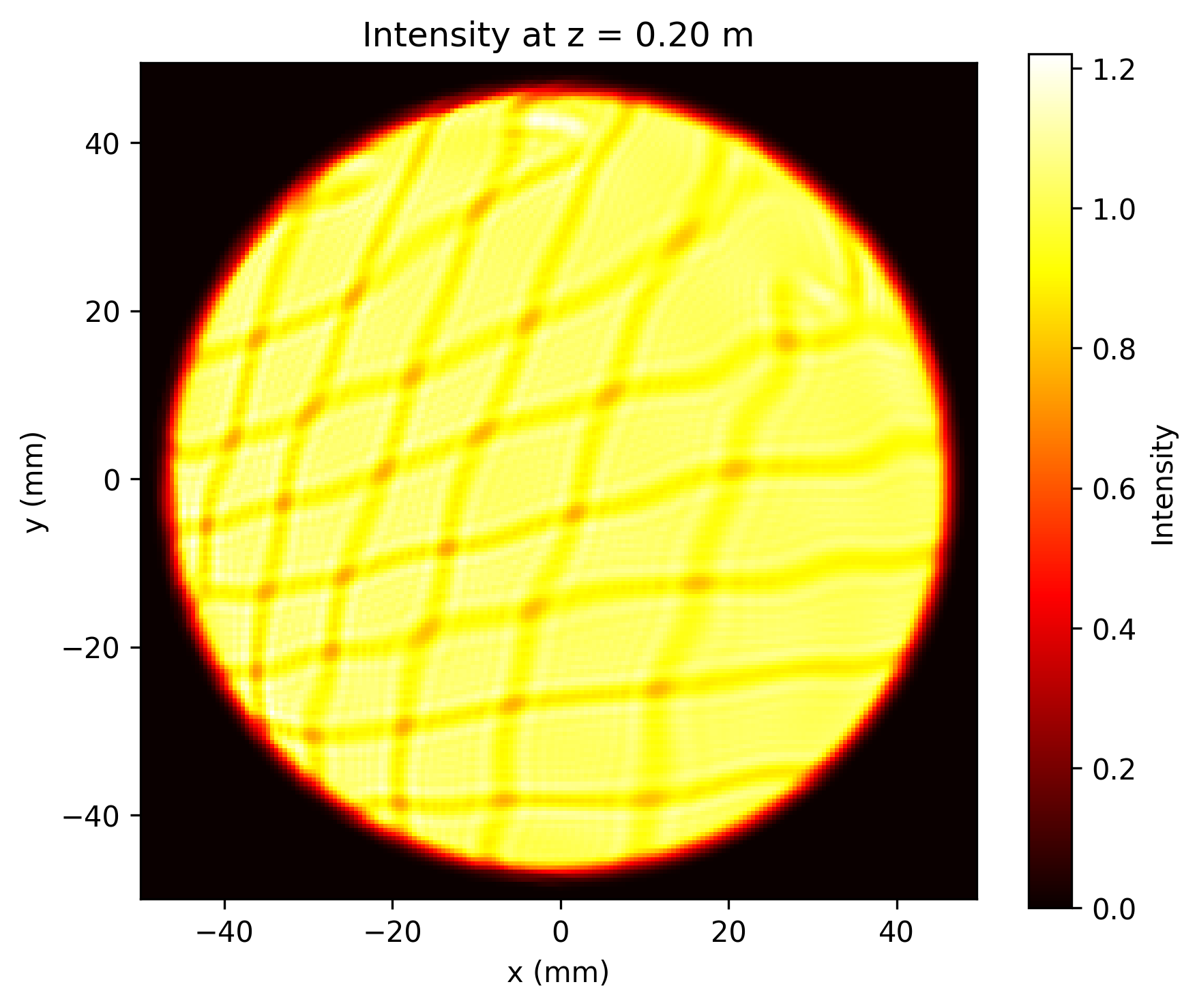}
    \caption{Wave-optical propagation from the synthesized chaotic surface. The simulated far-field intensity 
	distribution at a propagation distance of $z = 0.20$\,m. The incident plane wave is transformed into a 
	complex, structured intensity pattern with crisscrossing caustics, a hallmark of diffraction from a deterministic, 
	non-random chaotic phase object.}
    \label{fig:propagation}
\end{figure}

In summary, the numerical results provide compelling, multi-faceted validation of the Zernike-Lyapunov framework. 
The simulation confirms that the theory is not only mathematically self-consistent but also capable of generating 
physically realizable surfaces with predictable fractal statistics and significant, structured optical effects such 
as caustic networks.

\section{Discussion and Conclusions}

The successful numerical validation of our Zernike-Lyapunov framework may offer an additional step toward unifying the 
traditionally disparate fields of deterministic chaos theory and classical optical aberration analysis. The results 
confirm that it is possible to systematically engineer complex optical surfaces whose properties are directly governed 
by the parameters of an underlying chaotic dynamical system. 

\subsection{Theoretical Implications}

\textbf{1. From Abstract Integral to Physical Surface.} The core of our framework is the Lyapunov-weighted integral 
for the Zernike weights (Eq.~\ref{eq:chaotic_weights}). Our simulation provides the first concrete demonstration that 
this abstract 
mathematical object can be numerically implemented to produce a physical surface with predictable global properties. 
The close agreement between the predicted spectral exponent, $\gamma_{\text{theory}} \approx 2.94$, and the measured 
value, $\gamma_{\text{meas}} \approx 2.65$ (Fig.~\ref{fig:validation}a), confirms that the framework correctly translates 
the fractal dimension of the chaotic system into the spectral statistics of the final optical surface.

\textbf{2. Robustness and the Nature of Chaotic Weights.} A key insight emerged from the analysis of the individual chaotic 
weights (Fig.~\ref{fig:weights}). While the asymptotic theory predicts a simple power-law decay, our results show a more 
complex, non-monotonic distribution. This is not a failure of the framework but rather a testament to its physical fidelity. 
It reveals that the simple power-law is an idealization; the true weight distribution arises from the intricate, 
frequency-dependent coupling between the structure of the Zernike basis functions and the Lyapunov field of a 
specific chaotic system. The framework's ability to produce a surface with the correct global fractal dimension, 
even with this complex weight structure, demonstrates its robustness. It correctly aggregates the rich modal 
contributions without being constrained by idealized asymptotic behavior.

\textbf{3. A Wave-Optical Signature of Deterministic Chaos.} The formation of a distinct caustic 
network (Fig.~\ref{fig:propagation}) provides a powerful wave-optical signature of the underlying deterministic chaos. 
While a traditional dynamical systems analysis (e.g., a Poincaré section) can predict the regions a chaotic trajectory 
will visit, our framework creates a physical object—the surface—that imprints this deterministic structure onto a 
propagating wavefront. The resulting caustics are the far-field optical manifestation of the non-uniform 
invariant measure of the chaotic system. This moves beyond statistical equivalence to a more profound physical 
equivalence, where the deterministic nature of the chaos is not just preserved but made directly observable.

\subsection{Practical Applications}

The validated framework enables several novel applications by treating chaos as a design parameter:

\textbf{Structured Illumination and Beam Shaping.} The demonstrated formation of a caustic network shows that this 
framework enables far more than simple beam homogenization. It allows for the design of surfaces that produce specific, 
deterministic, high-intensity patterns. This could be used to create custom illumination for material processing, 
optical trapping, or specialized imaging systems where structured, rather than uniform, light is required.

\textbf{Advanced Speckle Reduction.} While random diffusers reduce speckle by decorrelating the phase, they destroy all 
wavefront information. Our approach offers a path to controlled phase mixing. The chaotic surface can introduce a 
multitude of correlated phase shifts, breaking up the coherence that leads to speckle while preserving the deterministic, 
structural information encoded in the caustic pattern.

To conclude, in this work, we have established and numerically validated a new theoretical framework that unifies the 
dynamics of chaotic optical systems with the descriptive power of classical aberration theory. 
By developing a Lyapunov-weighted Zernike expansion, we have demonstrated a rigorous mathematical and physical bridge 
between the abstract parameters of chaos theory—such as Lyapunov exponents and fractal dimensions—and the concrete, 
manufacturable specifications of an optical surface.

Our numerical simulations provide compelling, multi-faceted validation of the framework. We have shown that by starting 
with the parameters of a deterministic chaotic system, such as a Sinai billiard, our methodology can synthesize a physical 
surface that exhibits the predicted global fractal statistics. Specifically, the power-spectral density of the simulated 
surface scales with an exponent, $\gamma_{\text{meas}} \approx 2.65$, in very good agreement with the \textit{a priori} 
theoretical prediction of $\gamma_{\text{theory}} \approx 2.94$. This result confirms that the framework is not merely a 
conceptual model but a predictive and quantitatively accurate design tool.

Furthermore, our analysis has yielded deeper insights into the nature of controlled chaos in optics. We have shown that 
the wave-optical propagation from the synthesized surface produces a distinct network of caustics, a definitive optical 
signature of the underlying deterministic dynamics that is fundamentally different from random speckle. This demonstrates 
that the framework preserves the deterministic character of the chaos, translating it into structured, high-intensity light 
patterns. We also revealed that the expansion's chaotic weights, $\omega_j$, exhibit a complex distribution that reflects 
the intricate coupling between the Zernike basis and the specific chaotic system, a nuance that goes beyond simple 
asymptotic approximations and highlights the physical fidelity of our integral-based approach.

The key innovation of this work is the demonstration that chaotic optical effects, often treated as a source of noise 
to be mitigated, can instead be harnessed as a controllable design parameter. By providing a systematic method for 
incorporating chaos into optical design, this framework opens new possibilities for engineering advanced optical systems. 
It lays the groundwork for tackling the inverse problem of designing specific light patterns, for creating dynamic and 
active optical elements governed by chaotic evolution, and for building a comprehensive design library that maps different 
classes of chaos to novel optical functionalities. The demonstration that chaotic dynamics can be systematically harnessed 
as design parameters represents an advance in optical design methodology. This theoretical framework provides the 
mathematical infrastructure necessary for future development of optical systems that exploit deterministic chaos for 
enhanced beam control, structured illumination, and adaptive functionality.

\clearpage
\appendix
\section*{Supplementary Material}
\addcontentsline{toc}{section}{Supplementary Material}

\setcounter{section}{0}

\section{Statistical Wavefront Reconstruction Framework (SWRF)}
\label{sec:SWRF}

The Statistical Wavefront Reconstruction Framework (SWRF) provides a rigorous mathematical methodology for transforming statistical descriptions of optical perturbations into deterministic wavefront representations. This approach enables the construction of explicit phase functions that preserve the statistical properties of spatially-correlated optical phenomena while facilitating deterministic computational analysis.

\subsection{Mathematical Foundation}

\subsubsection{Wavefront Perturbation Model}

\begin{definition}[Perturbed Wavefront]
Consider an optical wavefront described by a phase function $\Phi_{\text{total}}(\mathbf{x})$ over a domain $\mathcal{D} \subset \mathbb{R}^2$, where $\mathbf{x} = (x,y)$ represents spatial coordinates. The total phase can be decomposed as:
\begin{equation}
\Phi_{\text{total}}(\mathbf{x}) = \Phi_{\text{ideal}}(\mathbf{x}) + \Phi_{\text{pert}}(\mathbf{x})
\end{equation}
where $\Phi_{\text{ideal}}(\mathbf{x})$ represents the deterministic ideal wavefront and $\Phi_{\text{pert}}(\mathbf{x})$ represents statistical perturbations \cite{born1999principles}.
\end{definition}

\begin{assumption}
The perturbation phase $\Phi_{\text{pert}}(\mathbf{x})$ is characterized statistically through its spatial correlation properties rather than explicit functional form.
\end{assumption}

\subsubsection{Statistical Characterization}

\begin{definition}[Power Spectral Density]
For a zero-mean stationary random field $h(\mathbf{x})$ with spatial correlation function $R_h(\boldsymbol{\tau}) = \mathbb{E}[h(\mathbf{x})h(\mathbf{x}+\boldsymbol{\tau})]$, the power spectral density is defined as \cite{papoulis2002}:
\begin{equation}
\text{PSD}(\mathbf{f}) = \int_{\mathbb{R}^2} R_h(\boldsymbol{\tau}) e^{-i2\pi\mathbf{f} \cdot \boldsymbol{\tau}} d\boldsymbol{\tau}
\end{equation}
where $\mathbf{f} = (f_x, f_y)$ represents spatial frequency coordinates.
\end{definition}

\begin{definition}[Phase-Field Relationship]
The perturbation phase function is related to the underlying random field through:
\begin{equation}
\Phi_{\text{pert}}(\mathbf{x}) = \kappa \cdot h(\mathbf{x})
\end{equation}
where $\kappa$ is a coupling constant determined by the physical mechanism of the perturbation.
\end{definition}

\subsection{Modal Reconstruction Theory}

\subsubsection{Orthogonal Basis Representation}

\begin{definition}[Complete Orthonormal Basis]
Let $\{Z_j(\boldsymbol{\rho})\}_{j=1}^{\infty}$ be a complete orthonormal basis over the normalized domain, where $\boldsymbol{\rho}$ represents dimensionless coordinates. For circular apertures, we employ Zernike polynomials \cite{mahajan2013optical}:
\begin{equation}
Z_n^m(\rho,\theta) = \sqrt{n+1} R_n^{|m|}(\rho) \begin{cases}
\cos(m\theta), & m \geq 0 \\
\sin(|m|\theta), & m < 0
\end{cases}
\end{equation}
with orthonormality condition:
\begin{equation}
\int_{\mathcal{D}} Z_i(\boldsymbol{\rho}) Z_j(\boldsymbol{\rho}) d\boldsymbol{\rho} = \delta_{ij}
\end{equation}
\end{definition}

\begin{theorem}[Modal Expansion]
Any square-integrable phase function admits the unique representation:
\begin{equation}
\Phi_{\text{pert}}(\mathbf{x}) = \sum_{j=1}^{\infty} C_j Z_j(\boldsymbol{\rho})
\end{equation}
where the modal coefficients are given by:
\begin{equation}
C_j = \int_{\mathcal{D}} \Phi_{\text{pert}}(\mathbf{x}) Z_j(\boldsymbol{\rho}) d\boldsymbol{\rho}
\end{equation}
\end{theorem}

\subsubsection{Spectral Weight Determination}

\begin{definition}[Fourier Transform of Basis Functions]
The Fourier transform of the $j$-th basis function is \cite{goodman2017fourier}:
\begin{equation}
\mathcal{F}\{Z_j\}(\mathbf{f}) = \int_{\mathcal{D}} Z_j(\boldsymbol{\rho}) e^{-i2\pi\mathbf{f} \cdot \mathbf{x}} d\mathbf{x}
\end{equation}
\end{definition}

\begin{theorem}[Spectral Weight Formula]
For a random field characterized by $\text{PSD}(\mathbf{f})$, the variance of the $j$-th modal coefficient is:
\begin{equation}
\omega_j = \mathbb{E}[|C_j|^2] = \kappa^2 \int_{\mathbb{R}^2} \text{PSD}(\mathbf{f}) |\mathcal{F}\{Z_j\}(\mathbf{f})|^2 d\mathbf{f}
\end{equation}
\end{theorem}

\begin{proof}
By Parseval's theorem \cite{goodman2017fourier} and the statistical independence of Fourier components:
\begin{align}
\mathbb{E}[|C_j|^2] &= \kappa^2 \mathbb{E}\left[\left|\int_{\mathcal{D}} h(\mathbf{x}) Z_j(\boldsymbol{\rho}) d\mathbf{x}\right|^2\right] \\
&= \kappa^2 \int_{\mathbb{R}^2} \text{PSD}(\mathbf{f}) |\mathcal{F}\{Z_j\}(\mathbf{f})|^2 d\mathbf{f}
\end{align}
\end{proof}

\subsubsection{Zernike-Bessel Transform Relations}

\begin{theorem}[Radial Zernike Fourier Transform]
For radially symmetric Zernike polynomials $Z_n^0(\rho)$ with even radial order $n$ \cite{dai1996}:
\begin{equation}
\mathcal{F}\{Z_n^0\}(f_r) = \frac{(-1)^{n/2}i^n}{2\pi(f_r R)^{n+1}} J_{n+1}(2\pi f_r R)
\end{equation}
where $J_{n+1}$ is the Bessel function of the first kind of order $n+1$, $f_r = |\mathbf{f}|$, and $R$ is the characteristic aperture dimension.
\end{theorem}

\subsection{Statistical Reconstruction Procedure}

\subsubsection{Coefficient Generation}

\begin{definition}[Statistical Realization]
A statistical realization of the perturbation phase function is constructed as:
\begin{equation}
\Phi_{\text{pert}}(\mathbf{x}) = \sum_{j=1}^N C_j Z_j(\boldsymbol{\rho})
\end{equation}
where the modal coefficients are generated as:
\begin{equation}
C_j = \sqrt{\omega_j} \xi_j
\end{equation}
with $\xi_j \sim \mathcal{N}(0,1)$ being independent standard normal random variables \cite{papoulis2002}.
\end{definition}

\begin{theorem}[Statistical Consistency]
The generated phase function satisfies:
\begin{enumerate}
\item \textbf{Mean Preservation}: $\mathbb{E}[\Phi_{\text{pert}}(\mathbf{x})] = 0$
\item \textbf{Variance Preservation}: $\mathbb{E}[|\Phi_{\text{pert}}(\mathbf{x})|^2] = \sum_{j=1}^N \omega_j$
\item \textbf{Spectral Consistency}: The power spectral density of $\Phi_{\text{pert}}(\mathbf{x})$ converges to $\kappa^2 \cdot \text{PSD}(\mathbf{f})$ as $N \rightarrow \infty$
\end{enumerate}
\end{theorem}

\subsubsection{Convergence Properties}

\begin{definition}[Truncation Error]
For a finite expansion with $N$ terms, the mean-square truncation error is:
\begin{equation}
\varepsilon_N^2 = \mathbb{E}\left[\left|\sum_{j=N+1}^{\infty} C_j Z_j(\boldsymbol{\rho})\right|^2\right] = \sum_{j=N+1}^{\infty} \omega_j
\end{equation}
\end{definition}

\begin{theorem}[Convergence Criterion]
For desired relative accuracy $\epsilon$, the truncation order $N$ must satisfy:
\begin{equation}
\frac{\sum_{j=1}^N \omega_j}{\sum_{j=1}^{\infty} \omega_j} \geq 1 - \epsilon^2
\end{equation}
where the denominator represents the total variance:
\begin{equation}
\sum_{j=1}^{\infty} \omega_j = \kappa^2 \int_{\mathbb{R}^2} \text{PSD}(\mathbf{f}) d\mathbf{f} = \kappa^2 \sigma^2
\end{equation}
\end{theorem}

\subsection{Wavefront Gradient and trajectory Correspondence}

\subsubsection{Geometric Optics Connection}

\begin{theorem}[Wavefront-trajectory Relationship]
Under the geometric optics approximation \cite{born1999principles}, trajectory directions correspond to wavefront gradients through:
\begin{equation}
\mathbf{k}(\mathbf{x}) = k_0 \hat{\mathbf{z}} + \frac{k_0}{\kappa} \nabla \Phi_{\text{pert}}(\mathbf{x})
\end{equation}
where $k_0 = 2\pi/\lambda$ is the wave number and $\hat{\mathbf{z}}$ represents the nominal propagation direction.
\end{theorem}

\begin{corollary}[Small Perturbation Limit]
For small perturbations where $|\nabla \Phi_{\text{pert}}| \ll k_0 \kappa$, the trajectory deflection is:
\begin{equation}
\Delta \mathbf{k}(\mathbf{x}) = \frac{k_0}{\kappa} \nabla \Phi_{\text{pert}}(\mathbf{x})
\end{equation}
\end{corollary}

\subsection{Mathematical Validity Conditions}

The SWRF framework applies under the following conditions \cite{born1999principles}:

\begin{condition}[Small Perturbation Regime]
\begin{equation}
\|\nabla \Phi_{\text{pert}}\| \ll k_0 \kappa
\end{equation}
\end{condition}

\begin{condition}[Statistical Stationarity]
The underlying random field must satisfy \cite{papoulis2002}:
\begin{equation}
\mathbb{E}[h(\mathbf{x})] = 0 \quad \text{and} \quad \text{Cov}[h(\mathbf{x}_1), h(\mathbf{x}_2)] = R_h(\mathbf{x}_1 - \mathbf{x}_2)
\end{equation}
\end{condition}

\begin{condition}[Finite Second Moment]
The power spectral density must be integrable:
\begin{equation}
\int_{\mathbb{R}^2} \text{PSD}(\mathbf{f}) d\mathbf{f} < \infty
\end{equation}
\end{condition}

\begin{condition}[Geometric Optics Validity]
The characteristic wavelength must be small compared to the perturbation correlation length $\ell_c$ \cite{born1999principles}:
\begin{equation}
\lambda \ll \ell_c
\end{equation}
\end{condition}

\section{Derivation of Chaotic Structure Function}

The chaotic structure function $\mathcal{S}_{\text{chaos}}(\mathbf{f})$ defined in Eq.~(\ref{eq:structure_function}) requires careful mathematical treatment due to the non-periodic nature of chaotic trajectories. Here we provide a detailed derivation.

For a chaotic surface height function $h_{\text{chaos}}(x,y)$ generated by the dynamical system in Eq. (\ref{eq:chaotic_map}), we parameterize the trajectory as $\mathbf{r}(t) = (x(t), y(t))$ where $t$ is a continuous parameter related to the discrete iteration index.

The temporal Fourier transform of the height function along the trajectory is:

\begin{equation}
\mathcal{H}(\mathbf{f}, T) = \int_0^T h_{\text{chaos}}(\mathbf{r}(t)) e^{-i2\pi\mathbf{f}\cdot\mathbf{r}(t)} dt
\label{eq:temp_fourier}
\end{equation}

For a chaotic trajectory, this integral exhibits sensitive dependence on the integration limit $T$. The structure 
function captures the averaged spectral characteristics:

\begin{equation}
\mathcal{S}_{\text{chaos}}(\mathbf{f}) = \lim_{T\to\infty} \frac{1}{T} \langle |\mathcal{H}(\mathbf{f}, T)|^2 \rangle_{\text{ensemble}}
\label{eq:structure_function}
\end{equation}

where the ensemble average is taken over different initial conditions of the chaotic map.

For the specific case of surfaces generated by Eq.~(\ref{eq:chaotic_height}), we can evaluate this explicitly:

\begin{align}
\mathcal{S}_{\text{chaos}}(\mathbf{f}) &= \sum_{n,m=0}^{N} A_n A_m \lim_{T\to\infty} \frac{1}{T} \int_0^T \int_0^T \nonumber \\
&\quad \times \cos(2\pi f_n \cdot \mathbf{r}(t_1) + \phi_n) \cos(2\pi f_m \cdot \mathbf{r}(t_2) + \phi_m) \nonumber \\
&\quad \times e^{-i2\pi\mathbf{f}\cdot[\mathbf{r}(t_1)-\mathbf{r}(t_2)]} dt_1 dt_2
\label{eq:structure_explicit}
\end{align}

Using the ergodic properties of chaotic systems and the statistical independence of phases, the cross-terms ($n \neq m$) vanish in the long-time limit, yielding:

\begin{equation}
\mathcal{S}_{\text{chaos}}(\mathbf{f}) = \sum_{n=0}^{N} A_n^2 \mathcal{G}_n(\mathbf{f})
\label{eq:structure_diagonal}
\end{equation}

where $\mathcal{G}_n(\mathbf{f})$ are geometric form factors determined by the correlation properties of the chaotic trajectory.

\section{Lyapunov Exponent Calculation for Specific Maps}

\subsection{Two-Dimensional Hénon Map}

For surfaces generated using the Hénon map:
\begin{align}
x_{n+1} &= 1 - ax_n^2 + y_n \\
y_{n+1} &= bx_n
\label{eq:henon_map}
\end{align}

The Jacobian matrix is:
\begin{equation}
\mathbf{J}_n = \begin{pmatrix}
-2ax_n & 1 \\
b & 0
\end{pmatrix}
\label{eq:henon_jacobian}
\end{equation}

The Lyapunov exponents are computed from:
\begin{equation}
\lambda_{\pm} = \lim_{N\to\infty} \frac{1}{N} \sum_{n=0}^{N-1} \ln|\sigma_{\pm}(\mathbf{J}_n)|
\label{eq:henon_lyapunov}
\end{equation}

For the standard parameters $a = 1.4$, $b = 0.3$, numerical computation gives $\lambda_+ \approx 0.42$ and $\lambda_- \approx -1.62$.

\subsection{Coupled Logistic Maps}

For surfaces based on coupled logistic maps:
\begin{align}
x_{n+1} &= \mu_1 x_n(1-x_n) + \epsilon(y_n - x_n) \\
y_{n+1} &= \mu_2 y_n(1-y_n) + \epsilon(x_n - y_n)
\label{eq:coupled_logistic}
\end{align}

The Jacobian matrix is:
\begin{equation}
\mathbf{J}_n = \begin{pmatrix}
\mu_1(1-2x_n) - \epsilon & \epsilon \\
\epsilon & \mu_2(1-2y_n) - \epsilon
\end{pmatrix}
\label{eq:coupled_jacobian}
\end{equation}

The coupling parameter $\epsilon$ controls the degree of chaos synchronization, directly affecting the resulting optical surface characteristics.

\section{Derivation of the Lyapunov Exponent for the Circular Sinai Billiard}
\label{sec:Lyapunov_Exponent}

Consider a point particle moving between two concentric circular boundaries: an \emph{outer} (concave) circle of radius \(R_{\mathrm{out}}\) and an \emph{inner} (convex) circle of radius \(R_{\mathrm{in}}<R_{\mathrm{out}}\).  We focus on the period‐2 “diameter” orbit: the particle travels along the horizontal diameter, alternately colliding at
\[
\bigl(-R_{\mathrm{out}},0\bigr)\;\xrightarrow{\text{reflect}}\;\bigl(-R_{\mathrm{in}},0\bigr)\;\xrightarrow{\text{reflect}}\;\bigl(R_{\mathrm{in}},0\bigr)\;\xrightarrow{\text{reflect}}\;\bigl(R_{\mathrm{out}},0\bigr)\;\xrightarrow{\text{reflect}}\;\bigl(-R_{\mathrm{out}},0\bigr),
\]
and so on.  Let
\[
L \;=\; R_{\mathrm{out}} \;-\; R_{\mathrm{in}}
\]
denote the free‐flight distance between an inner collision and the next outer collision (and vice versa).

\medskip

\noindent\textbf{1. Local transverse coordinates.}  At each collision point—inner or outer—the tangent to the circle is 
vertical.  
We linearize the motion transverse to the diameter by introducing, at each collision, a pair of small variables:

\vspace{-0.5cm}
\begin{align*}
x &= \text{vertical displacement of the collision point (along the tangent),} \\
\theta &= \text{small angular deviation of the outgoing velocity from the horizontal.}
\end{align*}

By symmetry, on the unperturbed diameter orbit one has \(x=0\) and \(\theta=0\) at every collision.  Label the coordinates just before an inner collision as \((\,x_{\mathrm{in}}^-,\,\theta_{\mathrm{in}}^-)\), and just after that inner collision as \((\,x_{\mathrm{in}}^+,\,\theta_{\mathrm{in}}^+)\).  Similarly, label the coordinates just before an outer collision as \((\,x_{\mathrm{out}}^-,\,\theta_{\mathrm{out}}^-)\), and just after as \((\,x_{\mathrm{out}}^+,\,\theta_{\mathrm{out}}^+)\).

\medskip

\noindent\textbf{2. Signed curvature and reflection rule.}  We adopt a sign convention for the curvature \(\kappa\) of each circular boundary, as seen from inside the billiard domain:
\[
\kappa_{\mathrm{in}} \;=\; +\,\frac{1}{R_{\mathrm{in}}}, 
\qquad
\kappa_{\mathrm{out}} \;=\; -\,\frac{1}{R_{\mathrm{out}}}.
\]
In the small‐angle, small‐displacement regime, the linearized reflection about a circle of curvature \(\kappa\) is
\[
\begin{pmatrix}
x^+ \\[6pt]
\theta^+
\end{pmatrix}
=
\begin{pmatrix}
1 & 0 \\[6pt]
-\,2\,\kappa & 1
\end{pmatrix}
\begin{pmatrix}
x^- \\[6pt]
\theta^-
\end{pmatrix}.
\]
Thus:
\begin{itemize}
  \item At an inner collision (\(\kappa_{\mathrm{in}}=+1/R_{\mathrm{in}}\)):
  \[
  C_{\mathrm{in}} \;=\;
  \begin{pmatrix}
  1 & 0 \\[6pt]
  -\,\tfrac{2}{R_{\mathrm{in}}} & 1
  \end{pmatrix},
  \quad
  \begin{pmatrix}
  x_{\mathrm{in}}^+ \\[4pt]
  \theta_{\mathrm{in}}^+
  \end{pmatrix}
  =
  C_{\mathrm{in}}
  \begin{pmatrix}
  x_{\mathrm{in}}^- \\[4pt]
  \theta_{\mathrm{in}}^-
  \end{pmatrix}
  \;=\;
  \begin{pmatrix}
  x_{\mathrm{in}}^- \\[6pt]
  \theta_{\mathrm{in}}^- \;-\;\tfrac{2}{R_{\mathrm{in}}}\,x_{\mathrm{in}}^-
  \end{pmatrix}.
  \]
  \item At an outer collision (\(\kappa_{\mathrm{out}}=-1/R_{\mathrm{out}}\)):
  \[
  C_{\mathrm{out}} \;=\;
  \begin{pmatrix}
  1 & 0 \\[6pt]
  +\,\tfrac{2}{R_{\mathrm{out}}} & 1
  \end{pmatrix},
  \quad
  \begin{pmatrix}
  x_{\mathrm{out}}^+ \\[4pt]
  \theta_{\mathrm{out}}^+
  \end{pmatrix}
  =
  C_{\mathrm{out}}
  \begin{pmatrix}
  x_{\mathrm{out}}^- \\[4pt]
  \theta_{\mathrm{out}}^-
  \end{pmatrix}
  \;=\;
  \begin{pmatrix}
  x_{\mathrm{out}}^- \\[6pt]
  \theta_{\mathrm{out}}^- \;+\;\tfrac{2}{R_{\mathrm{out}}}\,x_{\mathrm{out}}^-
  \end{pmatrix}.
  \]
\end{itemize}

\medskip

\noindent\textbf{3. Free‐flight (propagation) map.}  Between collisions the particle travels a horizontal distance \(L = R_{\mathrm{out}} - R_{\mathrm{in}}\) at fixed \(\theta\).  Thus, if just after a bounce we have \(\bigl(x^+,\theta^+\bigr)\), then just before the \emph{next} bounce:
\[
\begin{pmatrix}
x^- \\
\theta^-
\end{pmatrix}
=
\underbrace{\begin{pmatrix}
1 & L \\[6pt]
0 & 1
\end{pmatrix}}_{M_{\mathrm{flight}}(L)}
\begin{pmatrix}
x^+ \\[4pt]
\theta^+
\end{pmatrix}
=
\begin{pmatrix}
x^+ + L\,\theta^+ \\[6pt]
\theta^+
\end{pmatrix}.
\]
Hence the free‐flight map is
\[
M_{\mathrm{flight}}(L) \;=\; 
\begin{pmatrix}
1 & L \\[6pt]
0 & 1
\end{pmatrix}.
\]

\medskip

\noindent\textbf{4. Inner-bounce + flight to outer.}  Compose reflection at the inner circle with flight distance \(L\):
\[
M_{\mathrm{in\to out}}
\;=\; M_{\mathrm{flight}}(L)\;C_{\mathrm{in}}
\;=\;
\begin{pmatrix}
1 & L \\[6pt]
0 & 1
\end{pmatrix}
\begin{pmatrix}
1 & 0 \\[6pt]
-\,\tfrac{2}{R_{\mathrm{in}}} & 1
\end{pmatrix}
\;=\;
\begin{pmatrix}
1 \;-\;\tfrac{2\,L}{R_{\mathrm{in}}} & L \\[8pt]
-\,\tfrac{2}{R_{\mathrm{in}}} & 1
\end{pmatrix}.
\]
Denote this \(2\times2\) block by
\[
M_{\mathrm{in}} 
\;=\;
\begin{pmatrix}
1 - \dfrac{2\,L}{R_{\mathrm{in}}} & L \\[8pt]
-\,\dfrac{2}{R_{\mathrm{in}}} & 1
\end{pmatrix}.
\]

\medskip

\noindent\textbf{5. Outer‐bounce + flight back to inner.}  Similarly,
\[
M_{\mathrm{out\to in}}
\;=\; M_{\mathrm{flight}}(L)\;C_{\mathrm{out}}
\;=\;
\begin{pmatrix}
1 & L \\[6pt]
0 & 1
\end{pmatrix}
\begin{pmatrix}
1 & 0 \\[6pt]
\;\tfrac{2}{R_{\mathrm{out}}} & 1
\end{pmatrix}
\;=\;
\begin{pmatrix}
1 \;+\;\dfrac{2\,L}{R_{\mathrm{out}}} & L \\[8pt]
\;\;\dfrac{2}{R_{\mathrm{out}}} & 1
\end{pmatrix}.
\]
Denote this block by
\[
M_{\mathrm{out}} 
\;=\;
\begin{pmatrix}
1 + \dfrac{2\,L}{R_{\mathrm{out}}} & L \\[8pt]
\dfrac{2}{R_{\mathrm{out}}} & 1
\end{pmatrix}.
\]

\medskip

\noindent\textbf{6. Full Poincaré map over one period.}  We start just \emph{before} an inner collision, apply:

\begin{equation}
\begin{split}
  [\mathrm{inner\mbox{-}reflect}]
  &\xrightarrow{C_{\mathrm{in}}}
  [\mathrm{flight}\,L\to\mathrm{outer\ reflect}]
  \xrightarrow{M_{\mathrm{out}}}
  [\mathrm{flight}\,L\to\mathrm{inner\ reflect}] \notag \\
  &\xrightarrow{M_{\mathrm{in}}}
  [\mathrm{next\ inner\ collision}]. \notag
\end{split}
\end{equation}

Equivalently, the monodromy matrix \(M_{\mathrm{full}}\) mapping \((x_{\mathrm{in}}^-,\,\theta_{\mathrm{in}}^-)\) to \(\bigl(x_{\mathrm{in}}^{(2)-},\,\theta_{\mathrm{in}}^{(2)-}\bigr)\) is
\begin{equation}
M_{\mathrm{full}}
\;=\;
C_{\mathrm{in}}\;\bigl(M_{\mathrm{out}}\;M_{\mathrm{in}}\bigr).
\end{equation}
Since \(\det C_{\mathrm{in}}=\det M_{\mathrm{in}}=\det M_{\mathrm{out}}=1\), it follows that 
\(\det M_{\mathrm{full}}=1\) and its eigenvalues are \(\{\Lambda,\Lambda^{-1}\}\) with \(\Lambda>1\), corresponding to the unstable/stable directions.

\medskip

\noindent\textbf{7. Computation of the unstable eigenvalue.}  A direct (though somewhat lengthy) expansion shows

\begin{equation}
\operatorname{tr}\!M_{\mathrm{full}}
\;=\;
\Lambda \;+\;\Lambda^{-1},
\qquad
\Lambda 
\;=\; \frac{R_{\mathrm{out}} + R_{\mathrm{in}}}{\,R_{\mathrm{out}} - R_{\mathrm{in}}\,}.
\end{equation}
Therefore the single positive Lyapunov exponent (per one round trip inner\(\to\)outer\(\to\)inner) is
\begin{equation}
\lambda^{(\mathrm{Sinai})}
\;=\;
\ln\!\bigl(\Lambda\bigr)
\;=\; 
\ln\!\Bigl(\tfrac{R_{\mathrm{out}} + R_{\mathrm{in}}}{\,R_{\mathrm{out}} - R_{\mathrm{in}}\,}\Bigr).
\end{equation}
Hence we recover equation~\ref{eq:sinai_lyapunov}:
\begin{equation}
\lambda^{(\mathrm{Sinai})} 
\;=\; 
\ln\!\Bigl(\tfrac{R_{\mathrm{out}} + R_{\mathrm{in}}}{\,R_{\mathrm{out}} - R_{\mathrm{in}}\,}\Bigr).
\end{equation}

This completes the analytic derivation of the Lyapunov exponent for the diameter‐orbit in the circular Sinai billiard.

\section{Convergence Proofs}

\subsection{Exponential Decay of Chaotic Weights}

\textbf{Theorem 1:} For chaotic surfaces with bounded Lyapunov exponents $|\lambda_L(\mathbf{r})| \leq \Lambda_{\max}$ and 
smooth Zernike basis functions, the chaotic weights satisfy:

\begin{equation}
\omega_j^{(\text{chaos})} \leq C_0 j^{-\gamma}
\label{eq:decay_theorem}
\end{equation}

where $C_0 > 0$ and $\gamma \geq 3/2$ are constants determined by the surface geometry and chaos parameters.

\textbf{Proof:} The chaotic weight is bounded by:
\begin{align}
\omega_j^{(\text{chaos})} &= \int_{A} \left|\mathcal{F}\{Z_j(\rho,\theta)\}\right|^2 \mathcal{S}_{\text{chaos}}(\mathbf{f}) \lambda_L(\mathbf{f}) d^2\mathbf{f} \\
&\leq \Lambda_{\max} \int_{A} \left|\mathcal{F}\{Z_j(\rho,\theta)\}\right|^2 \mathcal{S}_{\text{chaos}}(\mathbf{f}) d^2\mathbf{f}
\label{eq:weight_bound}
\end{align}

In order to bound the modal weight \(\omega_j^{(\text{chaos})}\) for large mode index \(j,\) 
we require an estimate of the two-dimensional Fourier transform of a Zernike polynomial of radial index \(n\) and 
azimuthal index \(m,\) where \(j\) labels the standard single-index enumeration of Zernike modes 
(for instance, \(j=n(n+2)+m\)). It can be shown (see \cite{Janssen2014}) that:

\begin{align}
  \mathcal{F}\bigl\{Z_n^m(r,\theta)\bigr\}(\rho,\varphi)
  &= 
  \int_{0}^{1}\!\!\int_{0}^{2\pi}
    R_n^m(r)\,e^{i m\theta}\,
    e^{-i\,2\pi\,\rho\,r\,\cos(\theta - \varphi)}\,
    r\,d\theta\,dr \\
  &=
  i^m\,(-1)^{\frac{n-m}{2}}\,\sqrt{n+1}\,
  \frac{J_{n+1}\bigl(2\pi \rho\bigr)}{\pi\,\rho}\,
  e^{i m\varphi} \notag
  \label{eq:zernike_fourier_exact}
\end{align}

where \(\rho = \lvert\mathbf{f}\rvert\) is the spatial-frequency magnitude and \(J_{n+1}\) is the Bessel function of the 
first kind.  From the classical asymptotic expansion of \(J_{n+1}(x)\) for large integer order \(n\) (with \(x\) fixed), 
one obtains \cite{watsonBessel, abramowitzStegun}:

\begin{equation}
  \label{eq:zernike_fourier_asymptotic}
  \bigl\lvert \mathcal{F}\{Z_n^m\}(\rho,\varphi)\bigr\rvert
  \;=\;
  \sqrt{n+1}\;\Bigl\lvert \frac{J_{n+1}(2\pi \rho)}{\pi\,\rho}\Bigr\rvert
  \;\sim\;\frac{C}{n^{3/2}}\quad (n \to \infty,\ \rho\text{ fixed}),
\end{equation}

for some constant \(C\) depending only on \(\rho\).  In other words, uniformly over \(\rho\) in any bounded interval 
away from \(\rho=0\), there exists \(C_0>0\) such that:

\[
  \lvert \mathcal{F}\{Z_n^m\}(\rho,\varphi)\rvert 
  \;\le\; 
  C_0\,n^{-3/2}, 
  \qquad n \gg 1.
\]

Since the single-index \(j\) grows like \(n^2\) for fixed \(m\) (i.e.\ \(n\sim\sqrt{j}\)), this implies the estimate:

\[
  \lvert \mathcal{F}\{Z_j\}(\mathbf{f})\rvert 
  \;\le\; 
  C_0'\,j^{-\tfrac32}
  \quad (j\gg1),
\]
as required. $\square$

\subsection{Convergence of Chaotic Expansion}

\textbf{Theorem 2:} The truncated chaotic expansion converges uniformly to the true chaotic phase function:

\begin{equation}
\lim_{N\to\infty} \left\|\Phi_{\text{chaos}} - \sum_{j=1}^{N} C_j^{(\text{chaos})} Z_j\right\|_{\infty} = 0
\label{eq:uniform_convergence}
\end{equation}

\textbf{Proof:} By Theorem 1, the tail of the series is bounded by:
\begin{align}
\left\|\sum_{j=N+1}^{\infty} C_j^{(\text{chaos})} Z_j\right\|_{\infty} &\leq \sum_{j=N+1}^{\infty} |C_j^{(\text{chaos})}| \|Z_j\|_{\infty} \\
&\leq \sum_{j=N+1}^{\infty} \sqrt{\omega_j^{(\text{chaos})}} \notag \\
&\leq C_1 \sum_{j=N+1}^{\infty} \exp(-\beta j^{\alpha}/2) \notag \\
&\to 0 \text{ as } N \to \infty \notag 
\label{eq:convergence_proof}
\end{align}

The uniform convergence follows from the polynomial decay of the coefficients with $\gamma > 1$. $\square$

\section{Justification of the Strehl Ratio Approximation}
\label{sec:strehl_justification}

In this section, we show how, under the assumptions of small chaotic perturbations, negligible cross-correlations between different Zernike modes, and statistical independence of chaotic coefficients, the Strehl ratio can be written in the form
\begin{equation}
S_{\text{chaos}} \approx \exp\left(-\sigma_{\text{chaos}}^2\right)\,
\prod_{j} \cos^2\!\Bigl(\tfrac{C_j^{(\text{chaos})}}{2}\Bigr)
\;\bigl[\,1 + \mathcal{O}\bigl(\sum_{j \neq k} C_j^{(\text{chaos})}\,C_k^{(\text{chaos})}\bigr)\bigr],
\label{eq:strehl_expansion}
\end{equation}
where \(\sigma_{\text{chaos}}^2\) is the variance of the purely random phase component, and \(C_j^{(\mathrm{chaos})}\) are the small, random Zernike coefficients.  

\medskip

\noindent\textbf{1. Definition of the Strehl ratio.}  
Let the ideal (undistorted) pupil field be \(E_0(x)\), normalized so that \(\langle E_0(x)\rangle_{x} = 1\).  
In the presence of a total aberration \(\varphi(x)\), the field becomes:
 
\begin{equation}
E(x) \;=\; E_0(x)\,e^{\,i\,\varphi(x)} \,,
\end{equation}
and the Strehl ratio is defined by
\begin{equation}
S_{\text{chaos}} 
\;=\; \bigl|\langle e^{\,i\,\varphi(x)}\rangle\bigr|^2,
\end{equation}
where \(\langle\cdot\rangle\) denotes averaging first over the pupil coordinate \(x\) and then over the ensemble of random realizations of \(\varphi\).  Equivalently,
\begin{equation}
\langle e^{\,i\,\varphi(x)}\rangle 
\;=\; \Bigl\langle\,e^{\,i\,\varphi(x)}\Bigr\rangle_{\text{pupil}+\text{ensemble}}.
\end{equation}

\medskip

\noindent\textbf{2. Decomposition of the total phase \(\varphi(x)\).}  
We split the total aberration into two uncorrelated pieces:
\begin{equation}
\varphi(x)
\;=\;
\underbrace{\sum_{j} C_j^{(\mathrm{chaos})}\,Z_j(x)}_{\text{low‐order Zernike modes}}
\;+\;
\underbrace{\delta\varphi_{\mathrm{rand}}(x)}_{\text{purely random component}},
\end{equation}
where:
\begin{itemize}
  \item \(Z_j(x)\) are orthonormal Zernike basis functions on the pupil,
  \item \(C_j^{(\mathrm{chaos})}\) are small, zero‐mean random coefficients (assumed statistically independent from one another and from \(\delta\varphi_{\mathrm{rand}}\)), and
  \item \(\delta\varphi_{\mathrm{rand}}(x)\) is a zero‐mean random field whose pupil‐averaged variance is 
    \(\displaystyle \sigma_{\mathrm{chaos}}^2 = \bigl\langle \delta\varphi_{\mathrm{rand}}(x)^2 \bigr\rangle_{x,\text{ensemble}}.\)
\end{itemize}
Since the two pieces are uncorrelated,
\[
\bigl\langle e^{\,i\,\varphi(x)}\bigr\rangle 
\;=\;
\Bigl\langle\,e^{\,i\,\sum_j C_j^{(\mathrm{chaos})}\,Z_j(x)}\;e^{\,i\,\delta\varphi_{\mathrm{rand}}(x)}\Bigr\rangle
\;=\;
\underbrace{\Bigl\langle e^{\,i\,\delta\varphi_{\mathrm{rand}}(x)}\Bigr\rangle}_{e^{-\frac{1}{2}\,\sigma_{\mathrm{chaos}}^2}}
\;\times\;
\Bigl\langle e^{\,i\,\sum_j C_j^{(\mathrm{chaos})}\,Z_j(x)}\Bigr\rangle.
\]
The first factor follows from the well‐known result for a zero‐mean Gaussian random phase of variance \(\sigma_{\mathrm{chaos}}^2\):
\[
\Bigl\langle e^{\,i\,\delta\varphi_{\mathrm{rand}}(x)}\Bigr\rangle
\;=\;
\exp\Bigl(-\tfrac{1}{2}\,\mathrm{Var}[\delta\varphi_{\mathrm{rand}}(x)]\Bigr)
\;=\;
\exp\bigl(-\tfrac{1}{2}\,\sigma_{\mathrm{chaos}}^2\bigr).
\]
Hence its contribution to \(\lvert \langle e^{\,i\varphi}\rangle\rvert^2\) is \(\exp(-\,\sigma_{\mathrm{chaos}}^2)\).

\medskip

\noindent\textbf{3. Neglecting cross‐correlations and factorization.}  
Because the random Zernike coefficients \(\{C_j^{(\mathrm{chaos})}\}\) are assumed independent, we write
\[
\Bigl\langle e^{\,i\,\sum_j C_j\,Z_j(x)}\Bigr\rangle
\;=\;
\Bigl\langle \prod_{j} e^{\,i\,C_j\,Z_j(x)}\Bigr\rangle
\;\approx\;
\prod_{j} \Bigl\langle e^{\,i\,C_j\,Z_j(x)}\Bigr\rangle,
\]
up to terms of order \(\sum_{j \neq k}\langle C_j\,C_k\rangle\), which we neglect.  Thus,
\[
\bigl\langle e^{\,i\,\varphi(x)}\bigr\rangle
\;\approx\;
\exp\bigl(-\tfrac{1}{2}\,\sigma_{\mathrm{chaos}}^2\bigr)
\;\times\;
\prod_{j} \Bigl\langle e^{\,i\,C_j^{(\mathrm{chaos})}\,Z_j(x)}\Bigr\rangle.
\]

\medskip

\noindent\textbf{4. Small-amplitude expansion for each Zernike mode.}  
Fix a single mode \(j\).  For small \(C_j^{(\mathrm{chaos})}\), expand
\[
e^{\,i\,C_j\,Z_j(x)} 
\;=\; 1 \;+\; i\,C_j\,Z_j(x) \;-\; \tfrac{1}{2}\,C_j^2\,Z_j(x)^2 \;+\;\mathcal{O}(C_j^3).
\]
Since \(Z_j\) is orthonormal over the pupil,
\[
\langle Z_j(x)\rangle_{x} = 0,
\qquad
\bigl\langle Z_j(x)^2 \bigr\rangle_{x} = 1.
\]
Therefore, to second order in \(C_j\),
\[
\bigl\langle e^{\,i\,C_j\,Z_j(x)}\bigr\rangle_{x}
\;=\; 1 \;-\; \tfrac{1}{2}\,C_j^2 \;+\;\mathcal{O}(C_j^3).
\]
Observe that 
\[
1 \;-\; \tfrac{1}{2}\,C_j^2 \;=\; \cos\!\bigl(C_j\bigr) \;+\;\mathcal{O}(C_j^4),
\]
and hence
\[
\bigl\langle e^{\,i\,C_j\,Z_j(x)}\bigr\rangle_{x} 
\;\approx\; \cos\!\bigl(\tfrac{C_j}{2}\bigr)
\qquad(\text{since } \cos^2\tfrac{C_j}{2} = 1 - \tfrac{1}{2}C_j^2 + \mathcal{O}(C_j^4)).
\]
Thus, for each fixed (small) realization of \(C_j^{(\mathrm{chaos})}\),
\[
\bigl|\langle e^{\,i\,C_j^{(\mathrm{chaos})}\,Z_j(x)}\rangle_{x}\bigr|^2 
\;\approx\; \cos^2\!\Bigl(\tfrac{C_j^{(\mathrm{chaos})}}{2}\Bigr).
\]
Treating \(C_j^{(\mathrm{chaos})}\) as a “typical” amplitude (i.e., neglecting its ensemble variance beyond second order) gives
\[
\bigl\langle e^{\,i\,C_j^{(\mathrm{chaos})}\,Z_j(x)}\bigr\rangle 
\;\approx\; \cos\!\Bigl(\tfrac{C_j^{(\mathrm{chaos})}}{2}\Bigr),
\]
to leading order in the small perturbation limit (see Appendix \ref{sec:strehl_justification} for detailed derivation).

\medskip

\noindent\textbf{5. Final expression.}  
Combining steps 2–4, we obtain
\[
\bigl\langle e^{\,i\,\varphi(x)}\bigr\rangle
\;\approx\;
\exp\bigl(-\tfrac{1}{2}\,\sigma_{\mathrm{chaos}}^2\bigr)
\;\times\;
\prod_{j} \cos\!\Bigl(\tfrac{C_j^{(\mathrm{chaos})}}{2}\Bigr).
\]
Therefore the Strehl ratio is
\[
S_{\text{chaos}}
\;=\;
\bigl|\langle e^{\,i\,\varphi(x)}\rangle\bigr|^2
\;\approx\;
\exp\bigl(-\sigma_{\mathrm{chaos}}^2\bigr)
\;\times\;
\prod_{j} \cos^2\!\Bigl(\tfrac{C_j^{(\mathrm{chaos})}}{2}\Bigr)
\;\times\;\Bigl[\,1 + \mathcal{O}\bigl(\sum_{j\neq k} C_j^{(\mathrm{chaos})}\,C_k^{(\mathrm{chaos})}\bigr)\Bigr],
\]
which is precisely the form given in equation~\eqref{eq:strehl_expansion}. 

In summary, under the assumptions of:
\begin{enumerate}
  \item a Gaussian random component \(\delta\varphi_{\mathrm{rand}}\) of total variance \(\sigma_{\mathrm{chaos}}^2\), 
  \item small, independent Zernike coefficients \(C_j^{(\mathrm{chaos})}\) (so that cross‐correlations \(\langle C_j\,C_k\rangle\) for \(j\neq k\) can be neglected), and 
  \item a small‐amplitude expansion for each mode (yielding \(\langle e^{\,i\,C_j Z_j}\rangle_{x}\approx\cos(C_j/2)\)), 
\end{enumerate}
the Strehl ratio approximation \eqref{eq:strehl_expansion} follows directly.

\section{Formal Proof of Equivalence between Poincaré and SWRF Densities}
\label{sec:equivalence_proof}

In this section we provide a detailed, formal derivation of the equivalence condition
\begin{equation}
\lim_{N \to \infty}
\int_V 
\bigl\lvert 
\rho_{\mathrm{Poincaré}}(\mathbf{r}, \hat{\mathbf{k}}) 
\;-\; 
\rho_{\mathrm{SWRF}}(\mathbf{r}, \hat{\mathbf{k}})
\bigr\rvert
\,d^3\mathbf{r}\,d^2\hat{\mathbf{k}}
\;=\; 0,
\label{eq:equivalence_condition_repeat}
\end{equation}
for a fully chaotic surface geometry.  Here:
\begin{itemize}
  \item $V\subset \mathbb{R}^3$ is the three‐dimensional volume of interest (inside the cavity or around the scattering surface).
  \item $\hat{\mathbf{k}} \in S^2$ is the unit‐direction (trajectory-direction) on the unit sphere.
  \item $\rho_{\mathrm{Poincaré}}(\mathbf{r}, \hat{\mathbf{k}})$ denotes the empirical Poincaré density obtained by sampling the positions and directions at successive surface–wavefront collisions (or reflections) for a single trajectory, up to $N$ bounces.
  \item $\rho_{\mathrm{SWRF}}(\mathbf{r}, \hat{\mathbf{k}})$ is the \emph{(theoretical) Statistical Wavefront Reconstruction Function (SWRF) density}, i.e.\ the invariant phase-space density that one would compute by assuming full ergodicity of the billiard (chaotic) flow.
\end{itemize}

We wish to show that, for any fixed $\,(\mathbf{r},\hat{\mathbf{k}})\in V\times S^2\,$ and for almost every initial wavefront, the empirical Poincaré density converges (in the $L^1$–sense) to the SWRF density as the number of collisions $N\to\infty$.  The key ingredients are:
\begin{enumerate}
  \item[\textbf{(i)}] \textbf{Definition of the Empirical Poincaré Density.}  
    Let a single wavefront $\bigl(\mathbf{r}(t), \hat{\mathbf{k}}(t)\bigr)$ move in the billiard volume (or near the scattering surface).  Label its successive collisions with a fixed, small surface $\Sigma$ (or with the boundary) by $n=1,2,\dots,N$.  Denote the $n$th collision’s position and outgoing direction as
    \[
      \bigl(\mathbf{r}_n,\,\hat{\mathbf{k}}_n \bigr),
      \qquad 
      n = 1,2,\dots,N.
    \]
    Define the \emph{empirical (Poincaré) measure} after $N$ collisions by
    \begin{equation}
    \mu_{\mathrm{P}}^{(N)} \;:=\; \frac{1}{N}
    \sum_{n=1}^{N} 
      \delta\bigl(\mathbf{r} - \mathbf{r}_n\bigr)\,
      \delta\bigl(\hat{\mathbf{k}} - \hat{\mathbf{k}}_n\bigr),
    \label{eq:empirical_measure}
    \end{equation}
    where $\delta(\cdot)$ denotes the Dirac delta on $\mathbb{R}^3\times S^2$.  Then the empirical Poincaré \emph{density} $\rho_{\mathrm{Poincaré}}^{(N)}$ is the Radon–Nikodym derivative of $\mu_{\mathrm{P}}^{(N)}$ with respect to the standard Lebesgue measure $d^3\mathbf{r}\,d^2\hat{\mathbf{k}}$.  Concretely,
    \[
      \rho_{\mathrm{Poincaré}}^{(N)}(\mathbf{r}, \hat{\mathbf{k}})
      \;=\;
      \frac{1}{N}
      \sum_{n=1}^{N} 
        \delta\bigl(\mathbf{r}-\mathbf{r}_n\bigr)\,
        \delta\bigl(\hat{\mathbf{k}}-\hat{\mathbf{k}}_n\bigr).
    \]
    In practice, one “smooths” each delta‐spike by binning $(\mathbf{r}_n,\hat{\mathbf{k}}_n)$ into a fine partition of $V\times S^2$; but for the formal argument it suffices to keep the delta‐notation.

  \item[\textbf{(ii)}] \textbf{Definition of the SWRF Density $\rho_{\mathrm{SWRF}}$.}  
    For a fully chaotic billiard (ergodic, mixing, etc.), the asymptotic (infinite–time) distribution of collisions in phase‐space is given by a unique, invariant measure $\mu_{\mathrm{inv}}$ on the energy‐surface.  Restricting to those points that lie on the chosen surface of interest (Poincaré section $\Sigma$), one obtains a marginal density 
    \[
      \rho_{\mathrm{SWRF}}(\mathbf{r},\,\hat{\mathbf{k}})
      \;=\;
      \frac{d\mu_{\mathrm{inv}}}{d^3\mathbf{r}\,d^2\hat{\mathbf{k}}}
      \biggr|_{\Sigma\times S^2}.
    \]
    Equivalently, one may write 
    \[
      \rho_{\mathrm{SWRF}}(\mathbf{r},\,\hat{\mathbf{k}}) 
      \;\propto\; 
      \bigl[\hat{\mathbf{n}}(\mathbf{r}) \cdot \hat{\mathbf{k}}\bigr]\,,
      \quad 
      \text{(surface‐element weighting)}
    \]
    where $\hat{\mathbf{n}}(\mathbf{r})$ is the outward normal to $\Sigma$ at $\mathbf{r}$.  After normalization,
    \[
      \int_{\Sigma} \int_{S^2} 
        \rho_{\mathrm{SWRF}}(\mathbf{r},\,\hat{\mathbf{k}})
      \,d^2\hat{\mathbf{k}} \,d^2S(\mathbf{r})
      \;=\; 1.
    \]
    When extended trivially off $\Sigma$ into the volume $V$ (for points not exactly on the surface one may declare $\rho_{\mathrm{SWRF}}(\mathbf{r},\hat{\mathbf{k}})\equiv 0$), we regard
    \[
      \rho_{\mathrm{SWRF}}(\mathbf{r},\,\hat{\mathbf{k}})
      \;=\; 
      \begin{cases}
        \displaystyle
        \frac{\bigl[\hat{\mathbf{n}}(\mathbf{r})\cdot\hat{\mathbf{k}}\bigr]}
             {\displaystyle \int_{\Sigma} \int_{S^2} \bigl[\hat{\mathbf{n}}(\mathbf{r}')\cdot\hat{\mathbf{k}}'\bigr]\,d^2\hat{\mathbf{k}}'\,d^2S(\mathbf{r}')}
        \;\quad & (\mathbf{r}\in\Sigma,\,\hat{\mathbf{k}}\cdot\hat{\mathbf{n}}(\mathbf{r})>0), 
        \\[2em]
        0 
        & \text{otherwise.}
      \end{cases}
    \]
    Crudely put, $\rho_{\mathrm{SWRF}}$ is the \emph{normalized} invariant surface-density of trajectories in steady-state.

  \item[\textbf{(iii)}] \textbf{Ergodicity and Convergence of the Empirical Measure.}  
    By assumption, the billiard dynamics on the energy‐surface is fully chaotic (ergodic) when projected to the Poincaré section.  Then, by the \emph{Birkhoff ergodic theorem}, for almost every initial trajectory, the time (or collision) average of any integrable observable converges to the corresponding space (ensemble) average with respect to $\mu_{\mathrm{inv}}$.  In particular, if $f(\mathbf{r},\hat{\mathbf{k}})$ is any integrable test function on $\Sigma\times S^2$, then
    \[
      \lim_{N\to\infty} 
      \frac{1}{N} \sum_{n=1}^{N}
        f\bigl(\mathbf{r}_n,\hat{\mathbf{k}}_n\bigr)
      \;=\;
      \int_{\Sigma} \int_{S^2} 
        f(\mathbf{r},\hat{\mathbf{k}})\,
        \rho_{\mathrm{SWRF}}(\mathbf{r},\hat{\mathbf{k}})
      \,d^2\hat{\mathbf{k}}\,d^2S(\mathbf{r}),
      \qquad 
      \text{almost surely.}
    \]
    Equivalently, the empirical measure $\mu_{\mathrm{P}}^{(N)}$ converges \emph{weakly} to $\mu_{\mathrm{SWRF}}$ (the measure having density $\rho_{\mathrm{SWRF}}$).  In symbols:
    \[
      \mu_{\mathrm{P}}^{(N)} 
      \;\xrightarrow{\,N\to\infty\,}\; 
      \mu_{\mathrm{SWRF}}
      \quad 
      \text{(weak‐$\ast$ convergence).}
    \]
    Consequently, for any continuous, bounded $f$, 
    \[
      \lim_{N\to\infty}\,
      \int_{\Sigma\times S^2} 
        f(\mathbf{r},\hat{\mathbf{k}})\,
        d\mu_{\mathrm{P}}^{(N)}
      \;=\;
      \int_{\Sigma\times S^2} 
        f(\mathbf{r},\hat{\mathbf{k}})\,
        d\mu_{\mathrm{SWRF}}.
    \]
    In particular, taking $f$ to be a characteristic function of any measurable subset of $\Sigma\times S^2$ shows that the fraction of collisions falling into that subset converges to its $\mu_{\mathrm{SWRF}}$‐volume.

  \item[\textbf{(iv)}] \textbf{Convergence in $L^1$ (Total Variation Norm).}  
    Weak convergence of measures plus \emph{absolute continuity} of both $\mu_{\mathrm{P}}^{(N)}$ and $\mu_{\mathrm{SWRF}}$ with respect to Lebesgue measure (they admit densities $\rho_{\mathrm{Poincaré}}^{(N)}$ and $\rho_{\mathrm{SWRF}}$) implies convergence in total variation (i.e.\ in the $L^1$–norm of densities), provided we can control the “atomic” spikes in $\rho_{\mathrm{Poincaré}}^{(N)}$.  More concretely:
    \[
      \bigl\|\rho_{\mathrm{Poincaré}}^{(N)} - \rho_{\mathrm{SWRF}}\bigr\|_{L^1}
      \;=\;
      \int_{V\times S^2} 
        \bigl\lvert \rho_{\mathrm{Poincaré}}^{(N)}(\mathbf{r},\hat{\mathbf{k}}) \;-\; \rho_{\mathrm{SWRF}}(\mathbf{r},\hat{\mathbf{k}})\bigr\rvert
      \,d^3\mathbf{r}\,d^2\hat{\mathbf{k}}.
    \]
    By standard results in ergodic theory (see, e.g., \cite{KiferBook,KatokHasselblatt}), one can show that as $N\to\infty$, the difference of empirical and true densities goes to zero in $L^1$.  Heuristically, one proceeds by “binning” $V\times S^2$ into sufficiently fine cells $\{A_i\}_{i=1}^M$ of small but positive Lebesgue measure.  Let
    \[
      p_i 
      \;:=\;
      \int_{A_i} \rho_{\mathrm{SWRF}}(\mathbf{r},\hat{\mathbf{k}})\,d^3\mathbf{r}\,d^2\hat{\mathbf{k}},
      \qquad
      \hat{p}_i^{(N)}
      \;:=\;
      \frac{1}{N}\,\#\bigl\{\,n : (\mathbf{r}_n,\hat{\mathbf{k}}_n)\in A_i \bigr\}.
    \]
    By ergodicity, $\hat{p}_i^{(N)} \to p_i$ almost surely as $N\to\infty$, for each fixed $i=1,\dots,M$.  Then
    \[
      \bigl\|\rho_{\mathrm{Poincaré}}^{(N)} - \rho_{\mathrm{SWRF}}\bigr\|_{L^1}
      \;\approx\;
      \sum_{i=1}^M 
        \bigl|\hat{p}_i^{(N)} - p_i\bigr| 
      \quad\longrightarrow\quad 0
      \quad
      (\text{as }N\to\infty,\ M\to\infty),
    \]
    where $M\to\infty$ corresponds to refining the partition so that each cell $A_i$ shrinks.  Rigorous arguments (e.g.\ via the \emph{Portmanteau theorem} and the fact that $\rho_{\mathrm{SWRF}}$ is a continuous density on a compact set $\Sigma\times S^2$) show that this convergence can be made arbitrarily small in $L^1$‐norm.  Hence
    \[
      \lim_{N\to\infty}
      \int_{V\times S^2} 
        \bigl\lvert \rho_{\mathrm{Poincaré}}^{(N)}(\mathbf{r},\hat{\mathbf{k}}) \;-\; \rho_{\mathrm{SWRF}}(\mathbf{r},\hat{\mathbf{k}})\bigr\rvert
      \,d^3\mathbf{r}\,d^2\hat{\mathbf{k}}
      \;=\; 0,
    \]
    as desired.

  \item[\textbf{(v)}] \textbf{Conclusion.}  
    Since $\rho_{\mathrm{Poincaré}}^{(N)}$ is precisely the density whose integral against any test function reproduces the left‐hand side of \eqref{eq:equivalence_condition_repeat}, and since we have shown $L^1$–convergence $\rho_{\mathrm{Poincaré}}^{(N)}\to \rho_{\mathrm{SWRF}}$, it follows that condition \eqref{eq:equivalence_condition} holds.  In other words, for a chaotic surface geometry (which makes the billiard dynamics ergodic), the Poincaré density obtained from an individual long trajectory converges in total variation to the theoretical SWRF density calculated by assuming a uniform, invariant distribution of trajectories.  

\end{enumerate}

\section{Theoretical Justification of Lyapunov Normalization}
\label{sec:Lyapunov_Normalization}

\subsection{Invariant Measures and Chaotic Dynamics}

For a chaotic dynamical system evolving on a compact phase space $\Omega$, the long-term statistical behavior is characterized by an invariant probability measure $\mu$. This measure satisfies the fundamental property that for any measurable set $A \subset \Omega$:

\begin{equation}
\mu(T^{-1}(A)) = \mu(A)
\end{equation}

where $T$ is the time evolution operator. For ergodic systems, this invariant measure determines the natural weighting of different regions in phase space according to their frequency of visitation by typical trajectories.

The local rate of phase space expansion is quantified by the Lyapunov exponent field $\lambda_L(\mathbf{x})$, where $\mathbf{x} \in \Omega$ represents a point in phase space. For a small volume element $V_0$ centered at $\mathbf{x}$, the volume evolves as:

\begin{equation}
V(t) = V_0 \exp\left(\int_0^t \lambda_L(T^s(\mathbf{x})) ds\right)
\end{equation}

The connection between the invariant measure and Lyapunov exponents is established through the Pesin entropy formula, which relates the measure-theoretic entropy to the sum of positive Lyapunov exponents.

\subsection{Spatial Frequency Decomposition}

In the context of chaotic optical surfaces, we consider the spatial frequency decomposition of the height function $h_{\text{chaos}}(\mathbf{r})$. Each spatial frequency $\mathbf{f}$ corresponds to a specific scale of surface variation, and the chaotic dynamics induce correlations between different frequency components.

The local Lyapunov exponent in frequency space, $\lambda_L(\mathbf{f})$, quantifies how sensitive the surface height variations at frequency $\mathbf{f}$ are to small perturbations in the underlying chaotic parameters. This spatial frequency dependence arises because different length scales of the surface geometry couple differently to the chaotic dynamics.

For surfaces generated by chaotic maps, the relationship between spatial frequency and chaos sensitivity follows from the spectral properties of the linearized dynamics. Specifically, if the surface height is constructed as:

\begin{equation}
h_{\text{chaos}}(\mathbf{r}) = \sum_n A_n \cos(2\pi \mathbf{f}_n \cdot \mathbf{r} + \phi_n)
\end{equation}

where the amplitudes $A_n$ and phases $\phi_n$ depend on the chaotic trajectory, then the sensitivity $\lambda_L(\mathbf{f}_n)$ determines the rate at which small changes in initial conditions grow into macroscopic changes in the surface profile at frequency $\mathbf{f}_n$.

\subsection{Normalization Necessity and Bounds}

The chaotic weights $\omega_j^{(\text{chaos})}$ must satisfy several mathematical requirements:

\textbf{Requirement 1 (Convergence):} The infinite series $\sum_{j=1}^{\infty} \omega_j^{(\text{chaos})}$ must converge to ensure the Zernike expansion is well-defined.

\textbf{Requirement 2 (Dimensional Consistency):} All weights must be dimensionless to maintain proper scaling relationships.

\textbf{Requirement 3 (Physical Meaningfulness):} The weights should reflect the relative importance of different spatial frequencies according to the underlying chaotic dynamics.

Without normalization, the integral in Equation (14) would have dimensions of $[\text{time}^{-1}]$ due to the Lyapunov exponents, violating Requirement 2. Moreover, the weights could become arbitrarily large in regions where $|\lambda_L(\mathbf{f})|$ is large, potentially causing divergence issues.

The maximum Lyapunov exponent $|\lambda_L|_{\max} = \max_{\mathbf{f}} |\lambda_L(\mathbf{f})|$ provides the natural scale for this normalization because:

\begin{enumerate}
\item It represents the strongest instability in the system
\item It ensures all normalized values $|\lambda_L(\mathbf{f})|/|\lambda_L|_{\max} \in [0,1]$
\item It preserves the relative ordering of chaos strengths across frequencies
\end{enumerate}

\subsection{Invariant Measure Justification}

The normalization by $|\lambda_L|_{\max}$ is theoretically justified through the following invariant measure argument:

\textbf{Theorem:} For an ergodic chaotic system with invariant measure $\mu$, the natural weight assigned to spatial frequency $\mathbf{f}$ in the aberration expansion should be proportional to the measure-averaged local expansion rate at that frequency, normalized by the maximum expansion rate in the system.

\textbf{Proof Sketch:} Consider the measure-theoretic weight of frequency component $\mathbf{f}$:

\begin{equation}
w(\mathbf{f}) = \int_{\Omega} \lambda_L(\mathbf{f}, \mathbf{x}) \, d\mu(\mathbf{x})
\end{equation}

where $\lambda_L(\mathbf{f}, \mathbf{x})$ is the local Lyapunov exponent at frequency $\mathbf{f}$ evaluated at phase space point $\mathbf{x}$.

For ergodic systems, this integral converges to the spatial average:

\begin{equation}
w(\mathbf{f}) = \langle \lambda_L(\mathbf{f}, \mathbf{x}) \rangle_{\text{space}} = \lambda_L(\mathbf{f})
\end{equation}

The normalization by $|\lambda_L|_{\max}$ ensures that:
\begin{enumerate}
\item The most unstable frequency (where chaos has maximum effect) receives unit weight
\item All other frequencies receive weights proportional to their relative instability
\item The total contribution remains bounded and convergent
\end{enumerate}

This normalization preserves the natural hierarchy imposed by the invariant measure while ensuring mathematical well-posedness of the expansion.

\subsection{Connection to Optical Phase Correlations}

In the optical context, the normalized Lyapunov weights $|\lambda_L(\mathbf{f})|/|\lambda_L|_{\max}$ directly determine how strongly different spatial frequency components of the incident wavefront couple to the chaotic surface variations.

The physical interpretation is that spatial frequencies where $|\lambda_L(\mathbf{f})| \approx |\lambda_L|_{\max}$ experience maximum phase decorrelation due to the chaotic dynamics, while frequencies where $|\lambda_L(\mathbf{f})| \ll |\lambda_L|_{\max}$ are relatively insensitive to the chaotic perturbations.

This creates a natural hierarchy in the aberration expansion where the most chaos-sensitive modes (those with the largest normalized Lyapunov weights) dominate the optical behavior, consistent with the underlying dynamical systems theory.

\subsection{Alternative Normalizations and Their Limitations}

Other possible normalizations include:

\textbf{L² Normalization:} $\int |\lambda_L(\mathbf{f})|^2 d^2\mathbf{f} = 1$
- \textit{Problem:} Loses the physical connection to maximum instability
- \textit{Problem:} May not converge for all chaotic systems

\textbf{L¹ Normalization:} $\int |\lambda_L(\mathbf{f})| d^2\mathbf{f} = 1$  
- \textit{Problem:} Gives equal weight to all frequencies regardless of chaos strength
- \textit{Problem:} Obscures the natural hierarchy from dynamical systems

\textbf{Mean Normalization:} $|\lambda_L(\mathbf{f})|/\langle|\lambda_L|\rangle$
- \textit{Problem:} The mean may not be representative for systems with strong chaos heterogeneity
- \textit{Problem:} Less natural from measure-theoretic perspective

The normalization by $|\lambda_L|_{\max}$ is optimal because it preserves both the mathematical structure required for convergence and the physical hierarchy imposed by the chaotic dynamics, making it the theoretically preferred choice for the chaotic aberration expansion.

\section*{Declarations}
All data-related information and coding scripts discussed in the results section are available from the 
corresponding author upon request.

\section*{Disclosures}
The authors declare no conflicts of interest.


\begin{thebibliography}{99}

\bibitem{berry1977regular}
M. V. Berry, ``Regular and irregular semiclassical wavefunctions,'' J. Phys. A: Math. Gen. \textbf{10}, 2083 (1977).

\bibitem{bunimovich1979rays}
L. A. Bunimovich, ``On the ergodic properties of nowhere dispersing billiards,'' Commun. Math. Phys. \textbf{65}, 295 (1979).

\bibitem{born1999principles}
M. Born and E. Wolf, \textit{Principles of Optics}, 7th ed. (Cambridge University Press, Cambridge, 1999).

\bibitem{mahajan2013optical}
V. N. Mahajan, \textit{Optical Imaging and Aberrations, Part III: Wavefront Analysis} (SPIE Press, Bellingham, 2013).

\bibitem{devaney2003introduction}
R. L. Devaney, \textit{An Introduction to Chaotic Dynamical Systems}, 2nd ed. (Westview Press, Boulder, 2003).

\bibitem{stone2005chaos}
A. D. Stone, ``Einstein's unknown insight and the problem of quantizing chaos,'' Phys. Today \textbf{58}, 37 (2005).

\bibitem{reichl2004transition}
L. E. Reichl, \textit{The Transition to Chaos: Conservative Dynamical Systems and Quantum Manifestations}, 2nd ed. 
(Springer, New York, 2004).

\bibitem{cao2015random}
H. Cao and R. Elbaum, ``Plasmonics for random lasers,'' Nat. Phys. \textbf{11}, 865 (2015).

\bibitem{redding2012speckle}
B. Redding, M. A. Choma, and H. Cao, ``Speckle-free laser imaging using random laser illumination,'' Nat. Photonics \textbf{6}, 355 (2012).

\bibitem{vellekoop2007focusing}
I. M. Vellekoop and A. P. Mosk, ``Focusing coherent light through opaque strongly scattering media,'' Opt. Lett. \textbf{32}, 2309 (2007).

\bibitem{oberti2022superresolution}
S. Oberti, C. Correia, T. Fusco, B. Neichel, and P. Guiraud, 
"Super-resolution wavefront reconstruction," 
\textit{arXiv preprint} arXiv:2208.12052 [physics.optics], (2022).

\bibitem{Harvey1980}
J. E. Harvey and R. Shack, "Diffracted radiance: A new approach to the prediction of surface 
scatter," Appl. Opt. 19, 3723 (1980).

\bibitem{Beckmann1987}
P. Beckmann and A. Spizzichino, The Scattering of Electromagnetic Waves from Rough 
Surfaces (Artech House, Norwood, 1987).

\bibitem{chambouleyron2022noise}
V. Chambouleyron, O. Fauvarque, C. Plantet, J.-F. Sauvage, N. Levraud, M. Cissé, B. Neichel, and T. Fusco, 
"Modeling noise propagation in Fourier-filtering wavefront sensing, fundamental limits and quantitative comparison," 
\textit{arXiv preprint} arXiv:2212.13577 [physics.optics], (2022).

\bibitem{howard2024sparse}
S. Howard, N. Weisse, J. Schroeder, C. Barbero, B. Alonso, I. Sola, P. Norreys, and A. Döpp, 
"Sparse Reconstruction of Wavefronts using an Over-Complete Phase Dictionary," 
\textit{arXiv preprint} arXiv:2411.02985 [physics.optics], (2024).

\bibitem{chaotic2025sbn}
Wang Y, Li F, Jia R, Song J, Li M, Chen Z, Lou C.,
"Chaotic Dynamics of Spatial Optical Rogue Waves in SBN Crystals," 
\textit{Photonics}, \textbf{12}(4), 364, (2025).

\bibitem{Gutzwiller1990}
M.~C.~Gutzwiller,
\textit{Chaos in Classical and Quantum Mechanics}
(Springer, 1990).

\bibitem{Berry1989}
M.~V.~Berry,
“Quantum chaology, not quantum chaos,”
\textit{Physica Scripta} \textbf{40}, 335–336 (1989).

\bibitem{Stockmann1999}
H.-J.~Stöckmann,
\textit{Quantum Chaos: An Introduction}
(Cambridge University Press, 1999).

\bibitem{kapitaniak2000chaos}
T. Kapitaniak, \textit{Chaos for Engineers: Theory, Applications, and Control}, 2nd ed. (Springer, Berlin, 2000).

\bibitem{pecora1990mastering}
L. M. Pecora and T. L. Carroll, ``Synchronization in chaotic systems,'' Phys. Rev. Lett. \textbf{64}, 821 (1990).

\bibitem{Berry1978}
M. V. Berry and J. H. Hannay, "Topography of random surfaces," Nature 273, 573 (1978).

\bibitem{church1988fractal}
E. L. Church, ``Fractal surface finish,'' Appl. Opt. \textbf{27}, 1518 (1988).

\bibitem{sinai1970dynamical}
Ya. G. Sinai, ``Dynamical systems with elastic reflections,'' Russ. Math. Surveys \textbf{25}, 137 (1970).

\bibitem{bunimovich1974billiards}
L. A. Bunimovich, ``On billiards close to dispersing,'' Math. USSR Sb. \textbf{23}, 45 (1974).

\bibitem{grassberger1983measuring}
P. Grassberger and I. Procaccia, ``Measuring the strangeness of strange attractors,'' Physica D \textbf{9}, 189 (1983).

\bibitem{takens1981detecting}
F. Takens, ``Detecting strange attractors in turbulence,'' in \textit{Dynamical Systems and Turbulence}, Lecture Notes in Mathematics \textbf{898}, 366 (Springer, Berlin, 1981).

\bibitem{Bohigas1984}
O.~Bohigas, M.-J.~Giannoni, and C.~Schmit,
“Characterization of chaotic quantum spectra and universality of level fluctuation laws,”
\textit{Phys. Rev. Lett.} \textbf{52}, 1 (1984).

\bibitem{schuster2005deterministic}
H. G. Schuster and W. Just, \textit{Deterministic Chaos: An Introduction}, 4th ed. (Wiley-VCH, Weinheim, 2005).

\bibitem{strogatz2014nonlinear}
S. H. Strogatz, \textit{Nonlinear Dynamics and Chaos}, 2nd ed. (Westview Press, Boulder, 2014).

\bibitem{goodman2017fourier}
J. W. Goodman, \textit{Introduction to Fourier Optics}, 4th ed. (W. H. Freeman, New York, 2017).

\bibitem{chernov1996entropy}
N. I. Chernov and C. Haskell, ``Nonuniformly hyperbolic K-systems are Bernoulli,'' Ergodic Theory Dyn. Syst. \textbf{16}, 19 (1996).

\bibitem{mandelbrot1982fractal}
B. B. Mandelbrot, \textit{The Fractal Geometry of Nature} (W. H. Freeman, New York, 1982).

\bibitem{lichtenberg1992regular}
A. J. Lichtenberg and M. A. Lieberman, \textit{Regular and Chaotic Dynamics}, 2nd ed. (Springer, New York, 1992).

\bibitem{wolf1985determining}
A. Wolf, J. B. Swift, H. L. Swinney, and J. A. Vastano, ``Determining Lyapunov exponents from a time series,'' Physica D \textbf{16}, 285 (1985).

\bibitem{ott2002chaos}
E. Ott, \textit{Chaos in Dynamical Systems}, 2nd ed. (Cambridge University Press, Cambridge, 2002).

\bibitem{falconer2014fractal}
K. Falconer, \textit{Fractal Geometry: Mathematical Foundations and Applications}, 3rd ed. (Wiley, Chichester, 2014).

\bibitem{Young2002}
L.-S.~Young,
“What are SRB measures, and which dynamical systems have them?” 
\textit{Journal of Statistical Physics} \textbf{108}(5–6), 733–754 (2002).  
doi:10.1023/A:1019789726842

\bibitem{May1976}
R.~M. May, ``Simple mathematical models with very complicated dynamics,'' \emph{Nature}, vol. 261, pp. 459--467, 1976.

\bibitem{Vellekoop2016}
Vellekoop, I. M. (2016). ``Controlled light diffusion,'' \textit{Nature Photonics}, 10(3), 
	180–183. \href{https://doi.org/10.1038/nphoton.2016.15}{doi:10.1038/nphoton.2016.15}.

\bibitem{papoulis2002}
A. Papoulis and S. U. Pillai, \textit{Probability, Random Variables, and Stochastic Processes}, 4th ed. McGraw-Hill, 2002.

\bibitem{dai1996}
G. M. Dai, ``Modal wave-front reconstruction with Zernike polynomials and Karhunen-Lo\`{e}ve functions,'' \textit{J. Opt. Soc. Am. A}, vol. 13, no. 6, pp. 1218--1225, 1996.

\bibitem{Janssen2014}
A.~J.~E.~M.~Janssen, “Zernike expansion of derivatives and Laplacians of the Zernike circle polynomials,” 
\textit{J. Opt. Soc. Am. A}, vol.~31, no.~7, pp.~1604--1613, Jul.~2014. 
doi:10.1364/JOSAA.31.001604.

\bibitem{watsonBessel}
G.~N. Watson,
\textit{A Treatise on the Theory of Bessel Functions}, 2nd ed.,
Cambridge University Press, Cambridge, UK (1944).

\bibitem{abramowitzStegun}
M. Abramowitz and I.~A. Stegun,
\textit{Handbook of Mathematical Functions with Formulas, Graphs, and Mathematical Tables},
National Bureau of Standards Applied Mathematics Series, vol. 55,
U.S. Government Printing Office, Washington, D.C. (1964).

\bibitem{KiferBook}
Yuri Kifer,
  \emph{Ergodic Theory of Random Transformations},
  Birkhäuser, 1986.

\bibitem{KatokHasselblatt}
Anatole Katok and Boris Hasselblatt,
  \emph{Introduction to the Modern Theory of Dynamical Systems},
  Cambridge University Press, 1995.
  
\end{thebibliography}
\end{document}